\documentclass[aps,pra,twocolumn]{revtex4-2}
\usepackage{amssymb,amsmath,amsthm}
\usepackage{graphicx}
\usepackage{bm,bbm}
\usepackage[bookmarks=false]{hyperref}
\hypersetup{colorlinks=true,citecolor=blue,linkcolor=red,%
urlcolor=blue,pdfstartview=FitH,bookmarksopen=true}
\usepackage{comment}

\newcommand{\ket}[1]{\mbox{$ | #1 \rangle $}}
\newcommand{\bra}[1]{\mbox{$ \langle #1 | $}}
\newcommand{\braket}[2]{\mbox{$ \langle #1 | #2 \rangle $}}
\newcommand{\tr}{\mathrm{tr}}

\newcommand{\E}{\mathrm{e}}
\newcommand{\I}{\mathrm{i}}

\newtheoremstyle{note}
  {\topsep/2}              	
  {\topsep/2}            	
  {}                        
  {\parindent}             	
  {\itshape}                
  {.---}                    
  {0pt}                     
  {\thmname{#1}\thmnumber{ \itshape#2}\thmnote{ (#3)}} 

\newtheorem{theorem}{Theorem}
\newtheorem{lemma}{Lemma}

\theoremstyle{definition}

\theoremstyle{remark}

\begin{document}
\title{Efficient verification of entangled continuous-variable quantum states with local measurements}

\author{Ye-Chao Liu}
\affiliation{Key Laboratory of Advanced Optoelectronic Quantum Architecture and Measurement of
Ministry of Education, School of Physics, Beijing Institute of Technology, Beijing 100081, China}

\author{Jiangwei Shang}
\email{jiangwei.shang@bit.edu.cn}
\affiliation{Key Laboratory of Advanced Optoelectronic Quantum Architecture and Measurement of
Ministry of Education, School of Physics, Beijing Institute of Technology, Beijing 100081, China}

\author{Xiangdong Zhang}
\email{zhangxd@bit.edu.cn}
\affiliation{Key Laboratory of Advanced Optoelectronic Quantum Architecture and Measurement of
Ministry of Education, School of Physics, Beijing Institute of Technology, Beijing 100081, China}

\date{\today}
%

\begin{abstract}
Continuous-variable quantum states are of particular importance in various quantum information processing tasks including quantum communication and quantum sensing.
However, a bottleneck has emerged with the fast increasing in size of the quantum systems which severely hinders their efficient characterization.
In this work, we establish a systematic framework for verifying entangled continuous-variable quantum states by employing local measurements only.
Our protocol is able to achieve the unconditionally high verification efficiency which is quadratically better than quantum tomography as well as other nontomographic methods.
Specifically, we demonstrate the power of our protocol by showing the efficient verification of entangled two-mode and multimode coherent states with local measurements.
\end{abstract}

\maketitle
%

\section{Introduction}

Continuous-variable (CV) quantum systems have demonstrated their unique role in various quantum information processing applications \cite{Braunstein.etal2005, Weedbrook.etal2012}.
Because of the quantum description of the electromagnetic field, they are particularly relevant for
quantum communication and quantum-enhanced techniques including sensing, detecting, and imaging.
Also, the atomic and solid state CV systems have the potential for quantum computing.
The actual realization of all these applications must depend on the efficient and reliable characterization
of the quantum states in the first place.
The standard method of quantum state tomography (QST) \cite{DAriano.etal1994, Leonhardt1997, DAriano.etal2003,
Guta.etal2007, Glancy.etal2008, Lvovsky.etal2009, Donaldson.etal2015} is able to obtain all the information about the quantum states based on the Wigner function or the density matrix with a certain precision.
As powerful as it is, however, QST consumes too much resource, which is the reason why other nontomographic methods have been developed \cite{Aolita.etal2015, LiuNana.etal2019, Chabaud.etal2019, Chabaud.etal2020a, Chabaud.etal2021, Takeuchi.etal2019}.

Recently, a new characterization method called quantum state verification (QSV) has been systematically investigated in discrete-variable (DV) quantum systems \cite{Hayashi.etal2006, Pallister.etal2018}.
The task of QSV is to verify that a given quantum device does indeed produce a particular target state that we expect.
The core advantage of QSV lies in its asymptotically quadratic improvement of the verification efficiency as compared to other methods.
By the specific design, QSV can efficiently or even optimally verify many different kinds of multipartite quantum states with local measurements only \cite{Morimae.etal2017, Takeuchi.Morimae2018, Yu.etal2019, Li.etal2019, Wang.Hayashi2019, Zhu.Hayashi2019a, Zhu.Hayashi2019b, Zhu.Hayashi2019c, Zhu.Hayashi2019d, Liu.etal2019b, Li.etal2020b, Dangniam.etal2020, Zhang.etal2020a, Jiang.etal2020, Zhang.etal2020b, Li.etal2020a, Liu.etal2020b}, and the methodology can also be extended to verify quantum processes \cite{Liu.etal2020a,Zhu.Zhang2020,Zeng.etal2020}.
Considering CV quantum systems, however, even if the truncation method can reduce the dimension of the CV states from infinite to finite, the limited choice of measurements makes that the generalization of QSV to CV systems is, in general, hard.

To characterize CV quantum states, intensity measurements based on quadratures are usually employed in quantum tomography and other nontomographic methods.
The intensity measurements are especially suitable for Gaussian states, as they can be fully characterized by expectation values of the quadratic operators.
However, they are, in principle, inappropriate for the task of quantum verification since postprocessing of the experimental data is needed for estimating the quadratures.
Except for some special quantum states defined by quadratures directly like the CV cluster states \cite{Menicucci.etal2006}, of which the verification protocol can be generalized from the discrete scenario \cite{Takeuchi.etal2019}.
Hence, we consider the energy-based photon counting measurements, the realization of which relies on the single-photon detector (SPD) and the more general photon number resolution detector (PNRD) \cite{Mattioli.etal2016, Fukuda.etal2011, Divochiy.etal2008,  Kardynal.etal2008, Morais.etal2020}.
These measurements can realize projections on the Fock bases, as well as projections on the coherent states with the help of displacement operations.
Thus, postprocessing of the data can be avoided by using these ``deterministic'' measurements.
Note, in particular, different from the recent work by Wu \emph{et al.} \cite{Wu.etal2021} on the verification of CV quantum states which demands a necessary preprocessing step for the samples to satisfy the condition of independent and identical distribution, the intrinsic nature of QSV is in general exempt from this requirement.

In this work, we propose a systematic framework for verifying entangled CV quantum states with the help of local measurements only.
Our protocol is able to achieve the unconditionally high verification efficiency with the resource overhead given by ${N\propto O(\epsilon^{-1}\ln\delta^{-1})}$ within infidelity $\epsilon$ and confidence level ${1-\delta}$, which is quadratically better than quantum tomography as well as other nontomographic methods \cite{Aolita.etal2015, LiuNana.etal2019, Chabaud.etal2019, Chabaud.etal2020a, Chabaud.etal2021, Takeuchi.etal2019}.
For demonstration, we show the efficient verification of entangled two-mode as well as multimode coherent states with local measurements.
These states are crucial in various quantum information processing tasks including quantum teleportation \cite{vanEnk.etal2001, Xiaoguang2001}, quantum computation \cite{Cochrane.etal1999, Jeong.etal2002, Ralph.etal2003, Lund.etal2008}, and quantum metrology \cite{Joo.etal2011, Joo.etal2012, J.Liu.etal2016}.
Moreover, a general optimization strategy is given in order to achieve the optimal efficiency for the specific scenarios under consideration.

\section{General framework}

The task of quantum state verification is to determine whether the states ${\sigma_1,\sigma_2,\dots}$ output
from a device, all of which are supposed to be the target state \ket{\psi} (DV or CV), are cases either
${\sigma_i=\ket{\psi}\bra{\psi}}$ for all $i$, or ${\bra{\psi}\sigma_i\ket{\psi}\leq 1-\epsilon}$ for all $i$.
Then, a verification protocol $\Omega$ can be generally constructed by several dichotomic-outcome projective measurement settings $\{\Omega_l,\openone-\Omega_l\}$, such that
\begin{eqnarray}
  \Omega=\sum_l\mu_l\Omega_l\,,
\end{eqnarray}
where $\{\mu_l\}$ is a probability distribution.
With the requirement that the target state should always pass all the measurements, i.e., $\bra{\psi}\Omega_l\ket{\psi}=1$, errors of the verification protocol occur only when the noisy states $\sigma$ pass the protocol with the maximal probability \cite{Pallister.etal2018, Zhu.Hayashi2019c}
\begin{eqnarray}
  \max_{\langle\psi|\sigma|\psi\rangle\leq1-\epsilon} \tr(\Omega\sigma)=1-[1-\lambda_2(\Omega)]\epsilon=1-\nu\epsilon\,,
\end{eqnarray}
where $\lambda_2(\Omega)$ is the second-largest eigenvalue of $\Omega$, and
$\nu:=1-\lambda_2(\Omega)$ denotes the spectral gap from the maximal eigenvalue.
Then, after $N$ measurements, the maximal worst-case probability that the verifier fails to detect the ``bad'' case is given by $(1-\nu\epsilon)^N$.
Hence, to achieve a confidence level $1-\delta$, the number of copies of the states required is
\begin{eqnarray}\label{eq:number}
  N \geq \frac{1}{\ln[(1-\nu\epsilon)^{-1}]}\ln\delta^{-1} \approx \frac{1}{\nu}\epsilon^{-1}\ln\delta^{-1}\,.
\end{eqnarray}

In practice, searching for the optimal verification protocol is a demanding task, if not impossible at all.
For CV quantum states, in particular, the intrinsic nature of their infinite dimensional Hilbert space makes the spectral decomposition for $\Omega$ even more challenging.
To render the task, we may restrict the type of noisy states to be in some specific form.
Here, we focus on the noisy states such that ${\bra{\psi}\sigma_i\ket{\psi}=1-\epsilon}$ for all $i$,
which is allowable since other states with ${\bra{\psi}\sigma_i\ket{\psi}<1-\epsilon}$ would not make
the verification worse \cite{Pallister.etal2018}, thus the original task is retained.
By choosing a set of bases $\mathcal{S}$ constructed from the target state \ket{\psi} and all of the mutually orthonormal states $\{\ket{\psi_{i,j}^\perp}\}$ that form the subspace orthogonal to \ket{\psi},
we can write the noisy states as
\begin{eqnarray}\label{eq:noisy}
  \sigma_i = (1-\epsilon)\ket{\psi}\bra{\psi} +\sum_j\epsilon_{i,j}\ket{\psi_{i,j}^\perp}\bra{\psi_{i,j}^\perp}+{\rm N.D.}\,,
\end{eqnarray}
where $\sum_j\epsilon_{i,j}=\epsilon$ and ${\rm N.D.}$ represents the non-diagonal terms in $\mathcal{S}$,
i.e., $\ket{\phi'}\bra{\phi}$ for all ${\ket{\phi'}\neq\ket{\phi}}\in\mathcal{S}$.
Thus, the optimization for the spectral gap can be reduced to
\begin{eqnarray}\label{eq:v_opt}
     \nu_{\text{\rm opt}} := \max_{\Omega} \min_i \sum_l \mu_l k_{l,i}\,,
\end{eqnarray}
where ${k_{l,i}:=1-\sum_j\frac{\epsilon_{i,j}}{\epsilon}\bra{\psi_{i,j}^\perp}\Omega_l\ket{\psi_{i,j}^\perp} \in [0,1]}$;
see Appendix~A \cite{supp} for the detailed derivation.
Note that the above framework including the optimization method works for both the CV and DV scenarios.

\section{One-mode coherent state superpositions}

The one-mode coherent states are defined as
\begin{eqnarray}
  \ket{\alpha}=\E^{-\frac{|\alpha|^2}{2}}\sum_{n=0}^{\infty}\frac{\alpha^n}{\sqrt{n!}}\ket{n}\,,~~ \mbox{with}~\alpha\in\mathbb{C}\,,
\end{eqnarray}
where \ket{n} denotes the Fock states (or number states), which form a complete set of bases in the Hilbert space.
The verification of $\ket{\alpha}$ can be easily done by invoking the projective measurement $\ket{\alpha}\bra{\alpha}$, which is achievable via a Kennedy receiver \cite{Wittmann.etal2008, DiMario.etal2018},
i.e., a SPD combined with the displacement $D(-\alpha)$ in front,
\begin{eqnarray}\label{eq:Kenndy}
  \ket{\alpha}\bra{\alpha}\equiv D^{\dagger}(-\alpha)\ket{0}\bra{0}D(-\alpha)\,.
\end{eqnarray}

Then, we consider the superposition of one-mode coherent states, usually referred to as the coherent state superpositions (CSSs) \cite{Buzek.etal1992, Takahashi.etal2008}.
Typical examples of CSSs are the even and odd coherent cat states \cite{Dodonov.etal1974},
\begin{eqnarray}
  \ket{+_{\alpha}} = \frac{\ket{\alpha}+\ket{-\alpha}}{\sqrt{C_{+}^{(\alpha)}}}\,,\quad
  \ket{-_{\alpha}} = \frac{\ket{\alpha}-\ket{-\alpha}}{\sqrt{C_{-}^{(\alpha)}}}\,,
\end{eqnarray}
with the normalization $C_{\pm}^{(\alpha)}=2\bigl(1\pm \E^{{-2|\alpha|^2}}\bigr)$.
They are useful in various quantum information processing tasks including quantum teleportation \cite{vanEnk.etal2001}, quantum computation \cite{Ralph.etal2003, Lund.etal2008}, and quantum metrology \cite{Gilchrist.etal2004}.
These two cat states can be discriminated using the parity measurement \cite{Kuang.etal1996, Besse.etal2020}
\begin{eqnarray}
  \pi=\sum_n\ket{2n}\bra{2n}-\sum_n\ket{2n+1}\bra{2n+1}\,,
\end{eqnarray}
such that
\begin{equation}
\pi^{+}\ket{+_{\alpha}}=\ket{+_{\alpha}}\,,\quad
\pi^{-}\ket{-_{\alpha}}=\ket{-_{\alpha}}\,,
\end{equation}
which can be realized by PNRDs.
The superscripts $\pm$ indicate the projectors onto the
eigenspace with the corresponding eigenvalues $\pm 1$.
Note, however, that the parity measurement solely is not sufficient to verify the even nor odd coherent cat states, as for instance, \ket{+_{\alpha}} and \ket{+_{\beta}} with ${\alpha\neq\pm\beta}$ have the same behavior under the parity measurement.

Upon this point, we have revealed a significant difference between CV and DV quantum states regarding their verification, namely the CSSs including the symmetric cat states, cannot be verified straightforwardly.
The reason is due to the intrinsic nature of the infinite dimension of CV systems which results in the fact that realization of a deterministic arbitrary local operation in CV is still an open problem.
Nevertheless, some of the entanglement operations in CV are easier to implement instead.
For instance, in optical systems, consider the beam splitter (BS) with the form
\begin{eqnarray}
  B(\theta)=\exp\bigl[\I\theta\bigl(\hat{a}_1^{\dagger}\hat{a}_2+\hat{a}_1\hat{a}_2^{\dagger}\bigr)\bigr],
\end{eqnarray}
where $\hat{a}_{1(2)}^{\dagger},\hat{a}_{1(2)}$ are the creation and annihilation operators for the first (second) mode, respectively.
The parameter ${\theta\in(0,\frac{\pi}{2})}$ determines the transmissivity of the BS, i.e., ${\cos^2\theta\in(0,1)}$.
Then, by coupling to an ancilla vacuum mode \ket{0}, we can convert, for example, the even coherent cat state to
\begin{align}\label{eq:even_BS}
  B(\theta)\ket{+_{\alpha}}\ket{0}
  =
  \frac1{\sqrt{C_{+}^{(\alpha)}}}&\Bigl(\ket{\alpha\cos\theta}\ket{\I\alpha\sin\theta}\\[-5pt]
  &+\ket{-\alpha\cos\theta}\ket{-\I\alpha\sin\theta}\Bigr).\nonumber
\end{align}
This is the entangled coherent state, which we show how to verify in the next section.

\section{Two-mode entangled coherent states}

The entangled coherent states (ECSs) usually refer to the superposition of two-mode coherent states with the form $\ket{\alpha}\ket{0}+\ket{0}\ket{\alpha}$ (unnormalized) in specific \cite{Sanders1992,Sanders2012}.
Here, we consider a more general form such that
\begin{eqnarray}\label{eq:ECS_define}
  \ket{\psi^{\text{\rm ECS}}_{\alpha,\beta}} = \frac1{\sqrt{C_{\alpha,\beta}}}
  \Bigl(\ket{\alpha}\ket{0}+\ket{0}\ket{\beta}\Bigr),
\end{eqnarray}
where the normalization is $C_{\alpha,\beta}=2\bigl[1+\E^{-(|\alpha|^2+|\beta|^2)/2}\bigr]$.
Note that the transformed even coherent cat state in Eq.~\eqref{eq:even_BS} is a special case of $\ket{\psi^{\text{\rm ECS}}_{\alpha,\beta}}$ under proper local displacement operations.
Furthermore, a class of more general entangled states with the form (unnormalized)
\begin{eqnarray}\label{eq:two-mode}
  \ket{\Tilde{\psi}^{\text{\rm ECS}}}:=\ket{\alpha_1}\ket{\alpha_2}+\ket{\beta_1}\ket{\beta_2}
\end{eqnarray}
is locally equivalent to $\ket{\psi^{\text{\rm ECS}}_{\alpha,\beta}}$ when ${\alpha_{1,2}, \beta_{1,2}\in\mathbb{R}}$ \cite{Munro.etal2000,Jeong.etal2001}.
See Appendix~B \cite{supp} for all the derivations.

Although similar in form as the Bell states in the DV scenario, verification of the two-mode $\ket{\psi^{\text{\rm ECS}}_{\alpha,\beta}}$ cannot use the analytical technique as in Ref.~\cite{Li.etal2019} or the numerical approach as in Ref.~\cite{Yu.etal2019}.
The infinite dimensional Hilbert space of the CV systems severely restricts the type of measurements that can be realizable.
Hence, we start by fully exploring their physical properties.

Measurements of the SPDs can be described by $\{\tau^{+}=\ket{0}\bra{0},\tau^{-}=\openone -\ket{0}\bra{0}\}$, using which one finds the relation $(\tau^{-}\otimes\tau^{-})\ket{\psi^{\text{\rm ECS}}_{\alpha,\beta}}=0$.
Thus, to verify $\ket{\psi^{\text{\rm ECS}}_{\alpha,\beta}}$, the first measurement setting we can employ is
\begin{eqnarray}\label{eq:Omega1}
  \Omega_1^{\text{\rm ECS}} \!=\! \openone \!- \tau^{-}\!\otimes\tau^{-} \!=\! \openone \!-\! \sum_{n,m=1}^\infty \!\bigl(\ket{n}\bra{n}\otimes\ket{m}\bra{m}\bigr)\,,\quad
\end{eqnarray}
which satisfies $\Omega_1^\text{\rm ECS}\ket{\psi^{\text{\rm ECS}}_{\alpha,\beta}}=\ket{\psi^{\text{\rm ECS}}_{\alpha,\beta}}$.
Physically speaking, this setting implies that the target state must have at least one mode that has no photon.
However, the possibility that no photons emerge from both modes together cannot be ruled out since $\bra{\psi^{\text{\rm ECS}}_{\alpha,\beta}} (\tau^{+}\otimes\tau^{+}) \ket{\psi^{\text{\rm ECS}}_{\alpha,\beta}} \neq 0$, which is rather different from the property of the DV entangled states, like the Bell states.

Next, we define the displacement operation ${\mathbb{D}(-\alpha,-\beta):=D(-\alpha)\otimes D(-\beta)}$ acting on the two modes, which transforms the target state $\ket{\psi^{\text{\rm ECS}}_{\alpha,\beta}}$ to another ECS state $\ket{\psi^{\text{\rm ECS}}_{-\alpha,-\beta}}$.
With this, we have the following measurement setting
\begin{eqnarray}
  \Omega_2^{\text{\rm ECS}}=\mathbb{D}^{\dagger}\bigl(-\alpha,-\beta\bigr)\bigl(\openone \!- \tau^{-}\!\otimes\tau^{-}\bigr) \mathbb{D}\bigl(-\alpha,-\beta\bigr)\label{eq:Omega2}\,,
\end{eqnarray}
which is in fact a generalized Kennedy receiver for two-mode states.

The third measurement setting
\begin{eqnarray}\label{eq:Omega3}
  \Omega_3^{\text{\rm ECS}} = \mathbb{D}^{\dagger}\!\left(\!-\frac{\alpha}{2},-\frac{\beta}{2}\right)\!\bigl(\pi\otimes\pi\bigr)^{+} \mathbb{D}\!\left(\!-\frac{\alpha}{2},-\frac{\beta}{2}\right)
\end{eqnarray}
comes from the fact that ECSs have the parity symmetry under proper displacement, i.e.,
\begin{eqnarray}
  \mathbb{D}\!\left(\!-\frac{\alpha}{2},\!-\frac{\beta}{2}\right)\!\ket{\psi^{\text{\rm ECS}}_{\alpha,\beta}} \!=\!
  \frac{C_{\!+\!}
  \ket{+_{\frac{\alpha}{2}}}
  \ket{+_{\frac{\beta}{2}}}
  \!-\!
  C_{\!-\!}
  \ket{-_{\frac{\alpha}{2}}}
  \ket{-_{\frac{\beta}{2}}}}
  {2\sqrt{C_{\alpha,\beta}}}\,,\qquad
\end{eqnarray}
where $C_{\pm}=\sqrt{C_{\pm}^{(\alpha/2)}C_{\pm}^{(\beta/2)}}$.

\begin{figure}[t]
  \includegraphics[width=0.85\columnwidth]{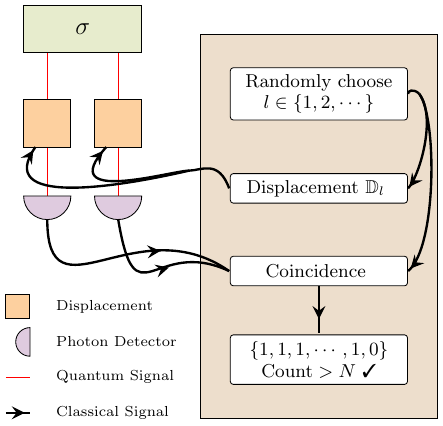}
  \caption{\label{fig:flow}General framework for the verification of two-mode entangled coherent states \ket{\psi^{\text{\rm ECS}}_{\alpha,\beta}} as in Eq.~\eqref{eq:ECS_define}.
  By randomly choosing the displacement operation $\mathbb{D}_l$, the outcomes with 1 or 0 of the coincidences of the two detectors are counted.
  The outcome 1 represents the ``pass'' instance and 0 for ``fail''.
  With the number of successive ``pass'' instances larger than $N$, \ket{\psi^{\text{\rm ECS}}_{\alpha,\beta}} is verified.
  }
\end{figure}

Briefly speaking, the first two measurement settings $\Omega_1^{\text{\rm ECS}}$ and $\Omega_2^{\text{\rm ECS}}$ check the existence of the vacuum mode and the coherent mode respectively.
Together they ensure the superposition of the two states $\ket{\alpha}\ket{0}$ and $\ket{0}\ket{\beta}$.
With the third setting $\Omega_3^{\text{\rm ECS}}$ confirming the balanced superposition, the ECSs can be verified.
Hence, these three measurement settings are sufficient to verify $\ket{\psi^{\text{\rm ECS}}_{\alpha,\beta}}$; see the following theorem with its proof postponed in Appendix~C \cite{supp}.
\begin{theorem}\label{thm:ECS}
  The two-mode entangled coherent states \ket{\psi^{\text{\rm ECS}}_{\alpha,\beta}} can be verified efficiently by the protocol
  \begin{eqnarray}
    \Omega^{\text{\rm ECS}}=\sum_{l=1}^3 \mu_l\Omega_l^{\text{\rm ECS}}\,,
  \end{eqnarray}
  where the probability distribution $\{\mu_l\}$ is arbitrary.
  An optimal efficiency can be obtained by optimizing $\{\mu_l\}$ under specific scenarios as constrained by Eq.~\eqref{eq:v_opt}.
\end{theorem}

In Fig.~\ref{fig:flow}, we show the general framework of the protocol.
The target states to be verified are input into two separate channels, followed by the displacement operation $\mathbb{D}_l$, and finally measured by photon detectors.
The displacement $\mathbb{D}_l$ has three different cases $\bigl\{\mathbb{D}_1\!=\!\openone\otimes\openone,\mathbb{D}_2\!=\!D(\!-\alpha)\otimes D(\!-\beta),\mathbb{D}_3\!=\!D(\!-\frac{\alpha}{2})\otimes D(\!-\frac{\beta}{2})\bigr\}$, which are applied on each mode locally.
PNRDs are required for the third setting $\Omega_3^{\text{\rm ECS}}$, while the first two settings only need SPDs.
The three measurement settings are randomly chosen in accordance with the probability distribution $\{\mu_l\}$.
Then, the numbers of coincidences of the two detectors are counted with outcome 1 representing a ``pass'' instance and 0 for ``fail''.
If the number of successive ``pass'' instances is larger than $N$, we confirm that the state is \ket{\psi^{\text{\rm ECS}}_{\alpha,\beta}} with certain confidence.
As mentioned by Theorem~\ref{thm:ECS}, the sample complexity $N$ is determined by the optimal verification efficiency $1/\nu_{\text{\rm opt}}$ which can be obtained by optimizing the probability distribution $\{\mu_l\}$ under specific scenarios as constrained by Eq.~\eqref{eq:v_opt}.
Note, furthermore, if changing the balanced superposition of the two terms in ECS from $+$ to $-$, to which the transformed odd coherent cat state is locally equivalent, all the experimental settings remain the same except for $\Omega_3^{\text{\rm ECS}}$ which differs by a sign; see Appendix~D \cite{supp} for more detailed discussions.

Before proceeding to give a concrete example, we have a quick remark regarding the PNRDs.
In practice, finite resolution of the PNRDs for the parity measurements always leads to a systematic error.
We can circumvent this problem by directly dismissing the results when PNRDs are saturated.
For instance, consider the PNRD$(3)$ with four outcomes $\{0,1,2,3+\}$, such that one can keep the outcomes $\{0,1,2\}$ only by discarding the rest.
Some efficiency will be lost during this process, and the probability to get the useful results with PNRD$(r)$ for a one-mode CV state \ket{\psi} is given by
\begin{eqnarray}
  p(r) = \sum_{i=0}^{r-1} |\bra{i}\psi\rangle|^2\,.
\end{eqnarray}
However, we emphasize that when $\alpha$ is small, the loss is negligible.
For example, considering the one-mode even coherent cat state \ket{+_{\alpha}} with $\alpha\!=\!1$, the efficiency is about 97.2\% by using PNRD$(3)$.
As for PNRD$(20)$ which is the highest resolution currently achievable in the laboratory \cite{Mattioli.etal2016}, the loss is around $10^{-16}$ which can be safely ignored.

Following the above discussion, here we demonstrate the high efficiency of our protocol by considering the verification of \ket{\psi^{\text{\rm ECS}}_{\alpha,\alpha}} with PNRD$(5)$.
The most significant decoherent noise for verifying ECSs comes from the photon loss from channels and the construction error from two-mode displacements.
Hence, as a demonstration, we assume that the input noisy states ${\sigma_i=\ket{\phi_i}\bra{\phi_i}}$ take the following forms, i.e.,
\begin{eqnarray}\label{eq:noisyStates}
  \ket{\phi_e(\eta)}&:=&\ket{\alpha-\kappa+\Delta}\ket{\kappa}
                        +\ket{\kappa+\Delta}\ket{\alpha-\kappa}\,,\\
  \ket{\phi_s(\eta)}&:=&\ket{\alpha-\kappa+\Delta}\ket{\kappa+\Delta}
                        +\ket{\kappa+\Delta}\ket{\alpha-\kappa+\Delta}\,.\nonumber
\end{eqnarray}
where the perturbation ${\kappa=(1-\sqrt{\eta})\alpha/2}$ is caused by photon loss with $\eta$ denoting the fraction of photons that survives the noisy channel \cite{vanEnk.etal2001}, and the small value $\Delta$ represents the displacement error.
These two types of noisy states thus correspond to the extreme one-mode error and the symmetric two-mode error respectively \footnote{Be reminded that the noisy states in Eq.~\eqref{eq:noisyStates} are different in form as the ones in Eq.~\eqref{eq:noisy}, but direct calculations can establish the relation between its passing probability and the infidelity $\epsilon$; see Appendix~E \cite{supp}}.
Then, the optimization in Eq.~\eqref{eq:v_opt} shows that the resource requirement for verifying \ket{\psi^{\text{\rm ECS}}_{\alpha,\alpha}} is
${N\approx 2.60(4)\epsilon^{-1}\ln\delta^{-1}}$ with the optimized probability
${\{\mu_l\}=\{0.49(7),0.40(2),0.10(1)\}}$.
As a comparison, for the tomographic detection of an ECS with a $99.99\%$ fidelity using PNRDs, the experiment in Ref.~\cite{Israel.etal2019} used more than $10^{10}$ measurements.
Our protocol, instead, is much more efficient which requires ${\sim\!10^5}$ measurements to reach the confidence level of $99\%$.
See Appendix~E \cite{supp} for more detailed discussions.

\section{multimode entangled coherent states}

Here we consider a class of multimode entangled coherent states with the form
\begin{eqnarray}\label{eq:GHZ}
  \ket{\psi^{\text{\rm GHZ-}m}}=\frac1{\sqrt{C}}\Biggl(\bigotimes_{i=1}^{m}\ket{\alpha_i}+\ket{0}^{\otimes m}\Biggr),
\end{eqnarray}
where the normalization is $C=2\bigl[1+\E^{-\sum_{i}|\alpha_i|^2/2}\bigr]$.
Depending on the number of modes, we refer to them as the $m$-mode GHZ-like coherent states,
which are generalizations of the states in Ref.~\cite{Jeong.etal2006}.
Note that, for $m=2$, states $\ket{\psi^{\text{\rm ECS}}_{\alpha,\beta}}$ and $\ket{\psi^{\text{\rm GHZ-}2}}$ are locally equivalent.
In fact, a class of generalized GHZ-like states
\begin{eqnarray}\label{eq:GGHZ}
  \ket{\Tilde{\psi}^{\text{\rm GHZ-}m}}:=\bigotimes_{i=1}^{m}\ket{\alpha_i}+\bigotimes_{i=1}^{m}\ket{\beta_i}
\end{eqnarray}
are locally equivalent to \ket{\psi^{\text{\rm GHZ-}m}} when $\alpha_i,\beta_i\in\mathbb{R}$ for all $i$; see Appendix B \cite{supp} for the derivations.

To verify $\ket{\psi^{\text{\rm GHZ-}m}}$, consider the following ${2(m\!-\!1)}$ measurement settings
\begin{eqnarray}\label{eq:GHZm_1}
  \Omega_{2l-1}^{\text{\rm GHZ-}m}&=&
  \mathcal{P}_{l}\biggl\{\!
  \Bigl[B_{2l\!-\!1}^\dagger\bigl(\openone \!- \tau^{-}\!\otimes\tau^{-}\bigr) B_{2l\!-\!1}\Bigr]\!
  \otimes\openone^{\otimes (m-2)}\!
  \biggr\},\nonumber\\
  \Omega_{2l}^{\text{\rm GHZ-}m}&=&
  \mathcal{P}_{l}\biggl\{\!
  \Bigl[B_{2l}^\dagger\bigl(\openone \!- \tau^{-}\!\otimes\tau^{-}\bigr) B_{2l}\Bigr]\!
  \otimes\openone^{\otimes (m-2)}\!
  \biggr\},
\end{eqnarray}
with $l\!=\!1,2,\cdots,m\!-\!1$, where ${B_{2l-1}=D(-\alpha_l)\otimes\openone}$ and $B_{2l}= \openone\otimes D(-\alpha_{l+1})$ are local operations.
The symbol $\mathcal{P}_{l}$ indicates the permutation that only the $l$ and $l+1$ modes are operated on for each setting.
Moreover, we need one more measurement setting
\begin{eqnarray}\label{eq:GHZm_2}
  \Omega_{2m-1}^{\text{\rm GHZ-}m}=\Biggl[\bigotimes_{i=1}^{m}D^{\dagger}\!\left(-\frac{\alpha_i}{2}\right)\Biggr]
  \bigl(\pi^{\otimes m}\bigr)^{+}
  \Biggl[\bigotimes_{i=1}^{m}D\!\left(-\frac{\alpha_i}{2}\right)\Biggr],\qquad
\end{eqnarray}
where $\bigl(\pi^{\otimes m}\bigr)^{+}$ is the projector onto the eigenspace with eigenvalue $+1$ of the parity measurement $\pi^{\otimes m}$.
With these, we have the following theorem for verifying $\ket{\psi^{\text{\rm GHZ-}m}}$; see Appendix~F \cite{supp} for the proof.
\begin{theorem}\label{thm:GHZ}
  The $m$-mode GHZ-like coherent states \ket{\psi^{\text{\rm GHZ-m}}} can be verified efficiently by the protocol
  \begin{eqnarray}
    \Omega^{\text{\rm GHZ-}m}=\sum_{l=1}^{2m-1}\mu_l\Omega_l^{\text{\rm GHZ-}m}\,,
  \end{eqnarray}
  where the probability distribution $\{\mu_l\}$ is arbitrary.
  An optimal efficiency can be obtained by optimizing $\{\mu_l\}$ under specific scenarios as constrained by Eq.~\eqref{eq:v_opt}.
\end{theorem}

Two remarks are in order.
First, except for the last setting which requires PNRDs on each mode, the other $2(m\!-\!1)$ measurement settings have exactly the same framework as $\Omega^{\text{\rm ECS}}$ shown in Fig~\ref{fig:flow}, thus only two modes are operated on each time.
Second, if changing the superposition of the two terms in \ket{\psi^{\text{\rm GHZ-}m}} from $+$ to $-$, they can still be verified efficiently with similar experimental settings; see Appendix~G \cite{supp} for the details.

\section{Conclusion}

We have developed a systematic framework for verifying continuous-variable quantum states with local measurements only.
Same as in the discrete-variable scenario, the high verification efficiency with the resource overhead of ${N\propto O(\epsilon^{-1}\ln\delta^{-1})}$ within infidelity $\epsilon$ and confidence level ${1-\delta}$ is retained.
This efficiency is quadratically better than quantum tomography as well as other nontomographic methods which usually require the resource in the order of ${N\propto O(\epsilon^{-2}\ln\delta^{-1})}$.
The high verification efficiency of our protocol is confirmed with the demonstration for verifying entangled two-mode and multimode coherent states.

As an outlook, it is interesting to extend the current study to verifying other important CV quantum states, and even CV quantum processes.
Also, the adversarial scenario \cite{Zhu.Hayashi2019c,Zhu.Hayashi2019d,Hayashi.Morimae2015,Takeuchi.etal2019} is worth considering, of which correlations among the input states can be included.
In such a case the scaling of the verification efficiency is expected to be kept, but deteriorate by a constant factor \cite{Zhu.Hayashi2019c,Zhu.Hayashi2019d}.
Moreover, with the techniques including adaptive measurements \cite{Yu.etal2019,Liu.etal2019b} and nondemolition photon counting \cite{Liu.etal2020b,Munro.etal2005,Guerlin.etal2007}, our protocol has the potential for further improvement.

\bigskip

\acknowledgments
We are grateful to Rui Han and Xiao-Dong Yu for helpful discussions.
This work was supported by the National Key R\&D Program of China under
Grant No.~2017YFA0303800 and the National Natural Science Foundation of
China through Grants No.~11574031, No.~61421001, and No.~11805010.


%

\onecolumngrid

\appendix

\section{Appendix A: Detailed derivation of the general framework}\label{app:framework}
We focus on the noisy states such that ${\bra{\psi}\sigma_i\ket{\psi}=1-\epsilon}$, which is allowable since other states with ${\bra{\psi}\sigma_i\ket{\psi}<1-\epsilon}$ would not make the verification worse \cite{Pallister.etal2018}.
By choosing a set of bases $\mathcal{S}$ constructed from the target state \ket{\psi} and all of the mutually orthonormal states $\{\ket{\psi_{i,j}^\perp}\}$ that form the subspace orthogonal to \ket{\psi}, we can write the noisy states as
\begin{eqnarray}\label{eq:def_noise}
  \sigma_i = \sum_{{ \ket{\phi},\ket{\phi'}}\in\mathcal{S}} \bra{\phi'}\sigma_i\ket{\phi} \ket{\phi'}\bra{\phi} = (1-\epsilon)\ket{\psi}\bra{\psi}+\sum_j \epsilon_{i,j}\ket{\psi_{i,j}^\perp}\bra{\psi_{i,j}^\perp}+\text{\rm N.D.}\,,
\end{eqnarray}
where ${\tr(\sigma_i)=1}$ indicates ${\sum_j\epsilon_{i,j}=\epsilon}$, and $\text{\rm N.D.}$ corresponds to the non-diagonal terms in $\mathcal{S}$.
With the constraint $\Omega_l\ket{\psi}=\ket{\psi}$, which implies $\bra{\psi^{\perp}_{i,j}}\Omega_l\ket{\psi}=0$, we have
\begin{eqnarray}
  \tr(\Omega_l\sigma_i)=1-\epsilon+\sum_j\epsilon_{i,j}\bra{\psi_{i,j}^\perp}\Omega_l\ket{\psi_{i,j}^\perp}\,.
\end{eqnarray}
Notice that $\bra{\psi_{i,j}^\perp}\Omega_l\ket{\psi_{i,j}^\perp}$ is the probability for the state \ket{\psi_{i,j}^\perp} to pass the measurement setting $\Omega_l$ which is in the range $[0,1]$, and the normalized parameters $\frac{\epsilon_{i,j}}{\epsilon}$ are determined by the structure of the noisy states $\sigma_i$.
Hence, the summation $\sum_j\frac{\epsilon_{i,j}}{\epsilon}\bra{\psi_{i,j}^\perp}\Omega_l\ket{\psi_{i,j}^\perp}$ is fixed in accordance with the measurement setting $\Omega_l$ and the specific type of noisy states $\sigma_i$ that we consider.
Next, define
\begin{eqnarray}
  k_{l,i}=1-\sum_j\frac{\epsilon_{i,j}}{\epsilon}\bra{\psi_{i,j}^\perp}\Omega_l\ket{\psi_{i,j}^\perp}\,,
\end{eqnarray}
which is thus independent of the infidelity $\epsilon$.
Following this definition, we have the linear relationship between the passing probability of the noisy state $\sigma_i$ and the infidelity $\epsilon$ such that
\begin{eqnarray}\label{eq:linear}
  \tr(\Omega_l\sigma_i)=1-k_{l,i}\epsilon\,,\quad k_{l,i}\in[0,1]\,.
  \end{eqnarray}
Then we have
\begin{eqnarray}
  \max_{\bra{\psi}\sigma_i\ket{\psi} = 1-\epsilon} \tr(\Omega\sigma) = 1-\min_i \sum_l \mu_l k_{l,i}\epsilon
  \quad \overset{i \to \infty}{=} \quad
  \max_{\bra{\psi}\sigma\ket{\psi}\leq 1-\epsilon}\tr(\Omega\sigma) = 1-\nu\epsilon
  \,.
\end{eqnarray}
Therefore, optimization of the spectral gap $\nu$ can be written as
\begin{eqnarray}
  \nu_{\text{\rm opt}}=\max_{\Omega}\min_i \sum_l \mu_l k_{l,i}\,,
  \end{eqnarray}
with the number of measurements required is
\begin{eqnarray}
  N \approx \frac{1}{\nu_{\text{\rm opt}}}\epsilon^{-1}\ln\delta^{-1}\,.
\end{eqnarray}
Note that if all the possible noisy states with ${\bra{\psi}\sigma_i\ket{\psi} = 1-\epsilon}$ are considered, an optimal upper bound for $N$ can be obtained.
However, if only considering some typical noisy states, the optimization can be simplified, which is often enough for practical implementations and the number of measurements is reduced accordingly.

\section{Appendix B: Local equivalence between several classes of CV states}
We first discuss the nature of the displacement operators in more details.
The displacement operator is also known as the shift operator
\begin{eqnarray}
  D(\alpha)=\E^{\alpha\hat{a}^\dagger+\alpha^{*}\hat{a}}\,,
\end{eqnarray}
where ${\alpha\in\mathbb{C}}$ represents the amount of displacement, and $\hat{a}^{\dagger},\hat{a}$ are the creation and annihilation operators.
The coherent states can be obtained by displacing the vacuum, i.e.,
\begin{eqnarray}
  D(\alpha)\ket{0}=\ket{\alpha}\,.
\end{eqnarray}
The product of two displacement operators is
\begin{eqnarray}
  D(\alpha)D(\beta)=\E^{(\alpha\beta^{*}-\alpha^{*}\beta)/2}D(\alpha+\beta)\,.
\end{eqnarray}
Thus, all of the one-mode coherent states are locally equivalent under displacement operations, as
\begin{eqnarray}
  \E^{(\alpha\beta^{*}-\alpha^{*}\beta)/2}=\E^{\I \Im(\alpha\beta^{*})}
\end{eqnarray}
is the global phase which is physically irrelevant,
and $\Im(\cdot)$ denotes the imaginary part of the number.

{\bf Case 1:}
Consider the balanced two-mode coherent state as defined in Eq.~\eqref{eq:two-mode},
\begin{eqnarray}
  \ket{\Tilde{\psi}^{\text{\rm ECS}}}=\frac1{\sqrt{\tilde{C}^{(2)}}}\Bigl(\ket{\alpha_1}\ket{\alpha_2}+\ket{\beta_1}\ket{\beta_2}\Bigr),
\end{eqnarray}
where the normalization is given by $\tilde{C}^{(2)}=2+\exp\!\left[-\frac{1}{2}\bigl(|\alpha_1|^2+|\beta_1|^2+|\alpha_2|^2+|\beta_2|^2\bigr)\right]\!\left(\E^{\alpha_1^{*}\beta_1+\alpha_2^{*}\beta_2}+\E^{\alpha_1\beta_1^{*}+\alpha_2\beta_2^{*}}\right)$.
With displacements operated on both modes, we have
\begin{eqnarray}
  D(-\beta_1)\otimes D(-\alpha_2)\ket{\Tilde{\psi}^{\text{\rm ECS}}}
  =
  \frac1{\sqrt{\tilde{C}^{(2)}}}\Bigl[\E^{(\alpha_1\beta_1^{*}-\alpha_1^{*}\beta_1)/2}
  \ket{\alpha_1-\beta_1}\ket{0}
  +\E^{(\alpha_2^{*}\beta_2-\alpha_2\beta_2^{*})/2}
  \ket{0}\ket{\beta_2-\alpha_2}\Bigr].
\end{eqnarray}
It is easy to find that $D(-\beta_1)\otimes D(-\alpha_2)\ket{\Tilde{\psi}^{\text{\rm ECS}}}$ and
\ket{\psi^{\text{\rm ECS}}_{\alpha_1-\beta_1,\beta_2-\alpha_2}} are equivalent whenever the constraint
\begin{eqnarray}\label{eq:cons1}
  \Im(\alpha_1\beta_1^{*})+\Im(\alpha_2\beta_2^{*})=2n\pi\,, \quad(n\in\mathbb{Z})
\end{eqnarray}
is satisfied,
of which real values of $\alpha_{1,2}$ and $\beta_{1,2}$ can always meet.

{\bf Case 2:}
We consider a more general form of the equivalence between the transformed even coherent cat state as in Eq.~\eqref{eq:even_BS} and the ECS.
The one-mode balanced CSSs can be written as \cite{Buzek.etal1992}
\begin{eqnarray}
  \ket{\psi^{\text{\rm BCSS}}}=\frac1{\sqrt{C^{\text{\rm BCSS}}}}\Bigl(\ket{\alpha}+\ket{\beta}\Bigr),
\end{eqnarray}
where the normalization is $C^{\text{\rm BCSS}}=2+\E^{-(|\alpha|^2+|\beta|^2)/2}\bigl(\E^{\alpha^{*}\beta}+\E^{\alpha\beta^{*}}\bigr)$.
By coupling to a vacuum mode, we utilize a beam splitter $B(\theta)$ and get
\begin{eqnarray}
  \ket{\Psi_{\alpha,\beta}}=\frac1{\sqrt{C^{\text{\rm BCSS}}}}\Bigl(\ket{\alpha\cos\theta}\ket{\I\alpha\sin\theta}+\ket{\beta\cos\theta}\ket{\I\beta\sin\theta}\Bigr).
\end{eqnarray}
Next, with the displacement $D(-\beta\cos\theta)\otimes D(-\I\alpha\sin\theta)$, it is converted to
\begin{eqnarray}
  \ket{\Psi_{\alpha,\beta}^{D}}
  =\frac1{\sqrt{C^{\text{\rm BCSS}}}}\Bigl[\E^{\frac{1}{2}(\alpha\beta^{*}-\alpha^{*}\beta)\cos^2\theta}
  \ket{(\alpha-\beta)\cos\theta}\ket{0}
  +\E^{\frac{1}{2}(\alpha^{*}\beta-\alpha\beta^{*})\sin^2\theta}
  \ket{0}\ket{\I(\beta-\alpha)\sin\theta}\Bigr].
\end{eqnarray}
Then, in order to make \ket{\Psi_{\alpha,\beta}^{D}} equivalent to the ECS, one needs the constraint
\begin{eqnarray}
  \Im(\alpha\beta^{*})=2n\pi\,, \quad(n\in\mathbb{Z})\,,
\end{eqnarray}
and the even coherent cat state is such a case with $\alpha=\beta$.

{\bf Case 3:}
The generalized multimode GHZ-like states are
\begin{eqnarray}
  \ket{\Tilde{\psi}^{\text{\rm GHZ-}m}}
  =\frac1{\sqrt{\tilde{C}^{(m)}}}\biggl(\bigotimes_{i=1}^{m}\ket{\alpha_i}+\bigotimes_{i=1}^{m}\ket{\beta_i}\biggr),
\end{eqnarray}
with the normalization given by $\tilde{C}^{(m)}=2+\exp\!\left[-\sum\bigl(|\alpha_i|^2+|\beta_i|^2\bigr)/2\right]\!\left[\exp\bigl(\sum\alpha_i^{*}\beta_i\bigr)+\exp\bigl(\sum\alpha_i\beta_i^{*}\bigr)\right]$.
With displacements, they are transformed to
\begin{eqnarray}
  \bigotimes_{i=1}^{m}D(-\beta_i)\ket{\Tilde{\psi}^{\text{\rm GHZ-}m}}
  =\frac1{\sqrt{\tilde{C}^{(m)}}}\biggl(\exp\Bigl[\sum_{i=1}^{m}(\alpha_i\beta_i^{*}-\alpha_i^{*}\beta_i)/2\Bigr]\bigotimes_{i=1}^{m}\ket{\alpha_i-\beta_i}+\ket{0}^{\otimes m}\biggr).
\end{eqnarray}
Then, if the constraint
\begin{eqnarray}\label{eq:cons-GHZ}
  \sum_i\Im(\alpha_i\beta_i^{*})=2n\pi\,, \quad (n\in\mathbb{Z})
\end{eqnarray}
is satisfied, \ket{\psi^{\text{\rm GHZ-}m}} and \ket{\Tilde{\psi}^{\text{\rm GHZ-}m}} are locally equivalent.
Similarly, real values of $\alpha_i,\beta_i$ for all $i$ can meet this requirement.

Additionally, we note that \ket{\psi_{\alpha,\beta}^{\text{\rm ECS}}} and \ket{\psi^{\text{\rm GHZ-}2}} are equivalent under the local displacement $\openone\otimes D(-\beta)$ for arbitrary $\alpha,\beta\in\mathbb{C}$, as the constraint in Eq.~\eqref{eq:cons-GHZ} is always satisfied.

\section{Appendix C: Proof of Theorem~\ref{thm:ECS}}\label{app:thm_ECS}
\begin{proof}
  What we want to prove is that the state which satisfies $\Omega_l\ket{\phi}=\ket{\phi}$ for all $i=1,2,3$ must take the form of ECSs, namely, $\ket{\phi}=\ket{\psi^{\text{\rm ECS}}_{\alpha,\beta}}$.

  First consider an arbitrary two-mode state $\ket{\phi}=\sum A_{kl}\ket{k}\ket{l}$.
  The photon counting measurement $\Omega_1^\text{\rm ECS}$ in Eq.~\eqref{eq:Omega1} can be rewritten as
  \begin{eqnarray}
    \Omega_1^\text{\rm ECS}
    &=&\sum_n \ket{n}\bra{n }\otimes\ket{0}\bra{0}
    +\ket{0}\bra{0}\otimes \sum_n \ket{n}\bra{n}
    -\ket{0}\bra{0}\otimes\ket{0}\bra{0}\nonumber\\
    &=&\openone\otimes\ket{0}\bra{0}+\ket{0}\bra{0}\otimes\openone-\ket{0}\bra{0}\otimes\ket{0}\bra{0}\,.
  \end{eqnarray}
  Then, the constraint $\Omega_1^\text{\rm ECS} \ket{\phi}=\ket{\phi}$ leads to
  \begin{eqnarray}\label{eq:phi_1}
    \ket{\phi}=\sum_{k=0}^\infty \bigl[A_{k0}\ket{k}\ket{0}+A_{0k}\ket{0}\ket{k}-A_{00}\ket{0}\ket{0}\bigr].
  \end{eqnarray}
  Similarly, the measurement $\Omega_2^\text{\rm ECS}$ in Eq.~\eqref{eq:Omega2} can be rewritten as
  \begin{eqnarray}
    \Omega_2^\text{\rm ECS}=\openone\otimes\ket{\beta}\bra{\beta}+\ket{\alpha}\bra{\alpha}\otimes\openone-\ket{\alpha}\bra{\alpha}\otimes\ket{\beta}\bra{\beta}\,,
  \end{eqnarray}
  and
  \begin{eqnarray}
    \Omega_2^\text{\rm ECS} \ket{\phi}
    &=&\sum_k \bigl[
      A_{k0}\ket{k}\ket{\beta}\langle\beta\ket{0}
      +A_{0k}\ket{0}\ket{\beta}\langle\beta\ket{k}
      -A_{00}\ket{0}\ket{\beta}\langle\beta\ket{0}
      \bigr]\nonumber\\
    &+&\sum_k \bigl[
      A_{k0}\ket{\alpha}\ket{0}\langle\alpha\ket{k}
      +A_{0k}\ket{\alpha}\ket{k}\langle\alpha\ket{0}
      -A_{00}\ket{\alpha}\ket{0}\langle\alpha\ket{0}
      \bigr]\nonumber\\
    &-&\sum_k \bigl[
      A_{k0}\ket{\alpha}\ket{\beta}\langle\beta\ket{0}\langle\alpha\ket{k}
      +A_{0k}\ket{\alpha}\ket{\beta}\langle\alpha\ket{0}\langle\beta\ket{k}
      -A_{00}\ket{\alpha}\ket{\beta}\langle\alpha\ket{0}\langle\beta\ket{0}
      \bigr].
  \end{eqnarray}
  Converting to the Fock bases, one gets
  \begin{eqnarray}
    \Omega_2^\text{\rm ECS}\ket{\phi}
    &=&\sum_{k,m}  \ket{k}\ket{m}A_{k0}\braket{m}{\beta}\braket{\beta}{0}
      +\sum_{k,m}  \ket{0}\ket{m}A_{0k}\braket{m}{\beta}\braket{\beta}{k}
      -\sum_{m}    \ket{0}\ket{m}A_{00}\braket{m}{\beta}\braket{\beta}{0}\nonumber\\
    &+&\sum_{k,n}  \ket{n}\ket{0}A_{k0}\braket{n}{\alpha}\braket{\alpha}{k}
      +\sum_{k,n}  \ket{n}\ket{k}A_{0k}\braket{n}{\alpha}\braket{\alpha}{0}
      -\sum_{n}    \ket{n}\ket{0}A_{00}\braket{n}{\alpha}\braket{\alpha}{0}\nonumber\\
    &-&\sum_{k,n,m}\ket{n}\ket{m}A_{k0}\braket{n}{\alpha}\braket{\alpha}{k}\braket{m}{\beta}\braket{\beta}{0}
      -\sum_{k,n,m}\ket{n}\ket{m}A_{0k}\braket{n}{\alpha}\braket{\alpha}{0}\braket{m}{\beta}\braket{\beta}{k}\nonumber\\
    &+&\sum_{n,m} \ket{n}\ket{m}A_{00}\braket{n}{\alpha}\braket{\alpha}{0}\braket{m}{\beta}\braket{\beta}{0}\,.
  \end{eqnarray}
  Here we agree on that Greek letters $\ket{\alpha}$ and $\ket{\beta}$ represent the coherent states, and the Fock states (number states) are represented by Latin letters $\ket{n}$,$\ket{m}$ and so on.
  Considering the amplitudes $A_{pq},A_{p0},A_{0q},A_{00}$ ($p\neq0$ and $q\neq0$), we have
  \begin{subequations}
    \begin{align}
      A_{pq} &= A_{p0}\braket{q}{\beta}\braket{\beta}{0}
              + A_{0q}\braket{p}{\alpha}\braket{\alpha}{0}
              - \sum_{k}A_{k0}\braket{p}{\alpha}\braket{\alpha}{k}\braket{q}{\beta}\braket{\beta}{0}\nonumber\\
             &\qquad\qquad\qquad
              - \sum_{k}A_{0k}\braket{p}{\alpha}\braket{\alpha}{0}\braket{q}{\beta}\braket{\beta}{k}
              + A_{00}\braket{p}{\alpha}\braket{\alpha}{0}\braket{q}{\beta}\braket{\beta}{0}\,,\label{eq:A_pq}\\
      A_{p0} &= A_{p0}|\braket{0}{\beta}|^2
              + \sum_{k}A_{k0}\braket{p}{\alpha}\braket{\alpha}{k}\bigl[1-|\braket{0}{\beta}|^2\bigr]
              - \sum_{k}A_{0k}\braket{p}{\alpha}\braket{\alpha}{0}\braket{0}{\beta}\braket{\beta}{k}
              + A_{00}\braket{p}{\alpha}\braket{\alpha}{0}|\braket{0}{\beta}|^2\,,\label{eq:A_p0}\\
      A_{0q} &= A_{0q}|\braket{0}{\alpha}|^2
              + \sum_{k}A_{0k}\braket{q}{\beta}\braket{\beta}{k}\bigl[1-|\braket{0}{\alpha}|^2\bigr]
              - \sum_{k}A_{k0}\braket{0}{\alpha}\braket{\alpha}{k}\braket{q}{\beta}\braket{\beta}{0}
              + A_{00}\braket{q}{\beta}\braket{\beta}{0}|\braket{0}{\alpha}|^2\,,\label{eq:A_0q}\\
      A_{00} &= A_{00}|\braket{0}{\alpha}|^2 |\braket{0}{\beta}|^2
              + \sum_{k}A_{0k}\braket{0}{\beta}\braket{\beta}{k}\bigl[1-|\braket{0}{\alpha}|^2\bigr]
              - \sum_{k}A_{k0}\braket{0}{\alpha}\braket{\alpha}{k}\bigl[1-|\braket{0}{\beta}|^2\bigr].\label{eq:A_00}
    \end{align}
  \end{subequations}
  Substituting Eqs.~\eqref{eq:A_p0}-\eqref{eq:A_00} into Eq.~\eqref{eq:A_pq}, together with the constraint $A_{pq}=0$ from Eq.~\eqref{eq:phi_1}, we have
  \begin{eqnarray}
    A_{pq} = A_{p0}\frac{\braket{q}{\beta}}{\braket{0}{\beta}}
           + A_{0q}\frac{\braket{p}{\alpha}}{\braket{0}{\alpha}}
           - A_{00}\frac{\braket{q}{\beta}\braket{p}{\alpha}}{\braket{0}{\beta}\braket{0}{\alpha}}
           = 0\,,
  \end{eqnarray}
  which is equivalent to
  \begin{eqnarray}
      A_{p0}\frac{\braket{0}{\alpha}}{\braket{p}{\alpha}}
    + A_{0q}\frac{\braket{0}{\beta}}{\braket{q}{\beta}}
    - A_{00}
    = 0\,.
  \end{eqnarray}
  As $A_{00}$ is independent of $p$ and $q$, the two terms $A_{p0}\frac{\braket{0}{\alpha}}{\braket{p}{\alpha}}$ and $A_{0q}\frac{\braket{0}{\beta}}{\braket{q}{\beta}}$ should be independent of $p$ and $q$ as well.
  Thus,
  \begin{eqnarray}
    B = A_{p0}\frac{\braket{0}{\alpha}}{\braket{p}{\alpha}}\,,\qquad
    C = A_{0q}\frac{\braket{0}{\beta}}{\braket{q}{\beta}}\,,\qquad
    B+C=A_{00}\,,\qquad(\forall B,C\in\mathbb{C})\,.
  \end{eqnarray}
  Hence, if \ket{\phi} is the eigenstate of both $\Omega_1^\text{\rm ECS}$ and $\Omega_2^\text{\rm ECS}$, it can be written as
  \begin{eqnarray}
    \ket{\phi}
    =\sum_{k=1}^{\infty}\frac{B}{\braket{0}{\alpha}}\underline{\braket{k}{\alpha}\ket{k}}\ket{0}
    +\sum_{k=1}^{\infty}\frac{C}{\braket{0}{\beta}}\ket{0}\underline{\braket{k}{\beta}\ket{k}}
    +\biggl(B\frac{\braket{0}{\alpha}}{\braket{0}{\alpha}}+C\frac{\braket{0}{\beta}}{\braket{0}{\beta}}\biggr)\ket{0}\ket{0}
    =\frac{B}{\braket{0}{\alpha}}\ket{\alpha}\ket{0}
     + \frac{C}{\braket{0}{\beta}}\ket{0}\ket{\beta}\,,
  \end{eqnarray}
  where we note that $\displaystyle\sum_{k=0}^\infty \braket{k}{\alpha}\ket{k}=\ket{\alpha}$, as all of the kets \ket{k} form a complete bases.
  With $B'=\frac{B}{\braket{0}{\alpha}}$ and $C'=\frac{C}{\braket{0}{\beta}}$, the state can be simply written as
  \begin{eqnarray}
    \ket{\phi}=B'\ket{\alpha}\ket{0}+C'\ket{0}\ket{\beta}\,.
  \end{eqnarray}
  Finally, considering the parity measurement $\Omega_3^\text{\rm ECS}$, we have
  \begin{eqnarray}
    \Omega_3^\text{\rm ECS}\ket{\phi}
    &=&\mathbb{D}^{\dagger}\!\left(\!-\frac{\alpha}{2},-\frac{\beta}{2}\!\right)\!
    \Bigl[
      B'\sum\ket{2n}\ket{2m}\braket{2n}{\frac{\alpha}{2}}\braket{2m}{-\frac{\beta}{2}}
      +B'\sum\ket{2n+1}\ket{2m+1}\braket{2n+1}{\frac{\alpha}{2}}\braket{2m+1}{-\frac{\beta}{2}}\nonumber\\
      &&\qquad~~~
      +C'\sum\ket{2n}\ket{2m}\braket{2n}{\!-\frac{\alpha}{2}}\braket{2m}{\frac{\beta}{2}}
      +C'\sum\ket{2n+1}\ket{2m+1}\braket{2n+1}{-\frac{\alpha}{2}}\braket{2m+1}{\frac{\beta}{2}}
      \Bigr]\nonumber\\
    &=&\mathbb{D}^{\dagger}\!\left(\!-\frac{\alpha}{2},-\frac{\beta}{2}\!\right)\!
    \Bigl[
      B'\bigl(\ket{\frac{\alpha}{2}}+\ket{-\frac{\alpha}{2}}\bigr)\bigl(\ket{\frac{\beta}{2}}+\ket{-\frac{\beta}{2}}\bigr)
      +B'\bigl(\ket{\frac{\alpha}{2}}-\ket{-\frac{\alpha}{2}}\bigr)\bigl(-\ket{\frac{\beta}{2}}+\ket{-\frac{\beta}{2}}\bigr)
      \nonumber\\
      &&\qquad\qquad
      +C'\bigl(\ket{\frac{\alpha}{2}}+\ket{-\frac{\alpha}{2}}\bigr)\bigl(\ket{\frac{\beta}{2}}+\ket{-\frac{\beta}{2}}\bigr)
      +C'\bigl(-\ket{\frac{\alpha}{2}}+\ket{-\frac{\alpha}{2}}\bigr)\bigl(\ket{\frac{\beta}{2}}-\ket{-\frac{\beta}{2}}\bigr)
      \Bigr]\times\frac{1}{4}\nonumber\\
    &=&\frac{B'+C'}{2}\big(\ket{\alpha}\ket{0}+\ket{0}\ket{\beta}\big),
  \end{eqnarray}
  and the constraint $\Omega_3^\text{\rm ECS}\ket{\phi}=\ket{\phi}$ gives
  \begin{eqnarray}
    B'=C'\,.
  \end{eqnarray}
  Therefore, the eigenstate of all the measurement settings $\Omega_1^\text{\rm ECS},\Omega_2^\text{\rm ECS},\Omega_3^\text{\rm ECS}$ is the entangled coherent state
  \begin{eqnarray}
    \ket{\phi}=\ket{\psi^{\text{\rm ECS}}_{\alpha,\beta}}=\frac{1}{\sqrt{2(1+e^{-|\alpha|^2})}}\Bigl[\ket{\alpha}\ket{0}+\ket{0}\ket{\beta}\Bigr].
  \end{eqnarray}
\end{proof}

\section{Appendix D: Verification of a different kind of ECS and its generalization}
Following the proof of Theorem~\ref{thm:ECS} in Appendix~C, here we present the verification protocol for a different kind of entangled coherent state with the form \cite{Sanders1992}
\begin{eqnarray}\label{eq:ECS2_define}
  \ket{{\psi^{\text{\rm ECS}(-)}_{\alpha,\beta}}} = \frac1{\sqrt{C_{\alpha,\beta}^{(-)}}}\Bigl(\ket{\alpha}\ket{0}-\ket{0}\ket{\beta}\Bigr),
\end{eqnarray}
where the normalization is $C_{\alpha,\beta}^{(-)}=2\bigl[1-\E^{-(|\alpha|^2+|\beta|^2)/2}\bigr]$.
This state is different from $\ket{\psi^{\text{\rm ECS}}_{\alpha,\beta}}$ in Eq.~\eqref{eq:ECS_define} by a sign in the phase.

Using two-mode generalized Kennedy receivers, the state \ket{{\psi^{\text{\rm ECS}(-)}_{\alpha,\beta}}} behaves exactly the same as that of \ket{{\psi^{\text{\rm ECS}}_{\alpha,\beta}}}.
The difference between these two states lies in their difference of the parity symmetry under displacements.
For \ket{{\psi^{\text{\rm ECS}(-)}_{\alpha,\beta}}}, we have
\begin{eqnarray}
  \mathbb{D}\!\left(\!-\frac{\alpha}{2},-\frac{\beta}{2}\right)\!\ket{{\psi^{\text{\rm ECS}(-)}_{\alpha,\beta}}} =
  \frac1{2\sqrt{C_{\alpha,\beta}^{(-)}}}\Bigl[C_{1}
  \ket{-_{\frac{\alpha}{2}}}
  \ket{+_{\frac{\beta}{2}}}
  -
  C_{2}
  \ket{+_{\frac{\alpha}{2}}}
  \ket{-_{\frac{\beta}{2}}}\Bigr],
\end{eqnarray}
where $C_{1}=\sqrt{C_{-}^{(\alpha/2)}C_{+}^{(\beta/2)}}$ and $C_{2}=\sqrt{C_{+}^{(\alpha/2)}C_{-}^{(\beta/2)}}$.
Thus, to verify \ket{\psi^{\text{\rm ECS}(-)}_{\alpha,\beta}}, the measurement settings are
\begin{eqnarray}
  {\Omega_1^{\text{\rm ECS}(-)}}
  &=&  \Omega_1^{\text{\rm ECS}}
  = \openone - \tau^{-}\otimes\tau^{-}\,,\\
  {\Omega_2^{\text{\rm ECS}(-)}}
  &=&  \Omega_2^{\text{\rm ECS}}
  =\mathbb{D}^{\dagger}(-\alpha,-\beta)(\openone - \tau^{-}\otimes\tau^{-}) \mathbb{D}(-\alpha,-\beta)\,,\\
  {\Omega_3^{\text{\rm ECS}(-)}}
  &=&  \mathbb{D}^{\dagger}\!\left(\!-\frac{\alpha}{2},-\frac{\beta}{2}\right)\bigl(\pi\otimes\pi\bigr)^{-} \mathbb{D}\!\left(\!-\frac{\alpha}{2},-\frac{\beta}{2}\right)\!.
\end{eqnarray}
See the following theorem with its proof.
\begin{theorem}\label{thm:ECS2}
  The two-mode entangled coherent states \ket{{\psi^{\text{\rm ECS}(-)}_{\alpha,\beta}}} can be verified efficiently by the protocol
  \begin{eqnarray}
    {\Omega^{\text{\rm ECS}(-)}}=\sum_{l=1}^3 \mu_l{\Omega_l^{\text{\rm ECS}(-)}}\,,
  \end{eqnarray}
  where the probability distribution $\{\mu_l\}$ is arbitrary.
  An optimal verification efficiency can be obtained by optimizing $\{\mu_l\}$ under specific scenarios as constrained by Eq.~\eqref{eq:v_opt}.
\end{theorem}
\begin{proof}
  The first part of the proof is the same as that in Appendix C, such that the eigenstate of both ${\Omega_1^{\text{\rm ECS}(-)}}$ and ${\Omega_2^{\text{\rm ECS}(-)}}$ can be written as
  \begin{eqnarray}
    \ket{\phi^{\prime}}=B\ket{\alpha}\ket{0}+C\ket{0}\ket{\beta}\,, \qquad(\forall B,C\in\mathbb{C})\,.
  \end{eqnarray}
  Then, considering the measurement $\Omega_3^{\text{\rm ECS}(-)}$, we have
  \begin{eqnarray}
    {\Omega_3^{\text{\rm ECS}(-)}}\ket{\phi^{\prime}}
    &=&\mathbb{D}^{\dagger}\!\left(\!-\frac{\alpha}{2},-\frac{\beta}{2}\!\right)\!
    \Bigl[
      B\sum\ket{2n}\ket{2m+1}\braket{2n}{\frac{\alpha}{2}}\braket{2m+1}{-\frac{\beta}{2}}
      +B\sum\ket{2n+1}\ket{2m}\braket{2n+1}{\frac{\alpha}{2}}\braket{2m}{-\frac{\beta}{2}}\nonumber\\
      &&\qquad~~~
      +C\sum\ket{2n}\ket{2m+1}\braket{2n}{-\frac{\alpha}{2}}\braket{2m+1}{\frac{\beta}{2}}
      +C\sum\ket{2n+1}\ket{2m}\braket{2n+1}{-\frac{\alpha}{2}}\braket{2m}{\frac{\beta}{2}}
      \Bigr]\nonumber\\
    &=&\mathbb{D}^{\dagger}\!\left(\!-\frac{\alpha}{2},-\frac{\beta}{2}\right)\!
    \Bigl[
      B\bigl(\ket{\frac{\alpha}{2}}+\ket{-\frac{\alpha}{2}}\bigr)\bigl(-\ket{\frac{\beta}{2}}+\ket{-\frac{\beta}{2}}\bigr)
      +B\bigl(\ket{\frac{\alpha}{2}}-\ket{-\frac{\alpha}{2}}\bigr)\bigl(\ket{\frac{\beta}{2}}+\ket{-\frac{\beta}{2}}\bigr)
      \nonumber\\
      &&\qquad\qquad
      +C\bigl(\ket{\frac{\alpha}{2}}+\ket{-\frac{\alpha}{2}}\bigr)\bigl(\ket{\frac{\beta}{2}}-\ket{-\frac{\beta}{2}}\bigr)
      +C\bigl(-\ket{\frac{\alpha}{2}}+\ket{-\frac{\alpha}{2}}\bigr)\bigl(\ket{\frac{\beta}{2}}+\ket{-\frac{\beta}{2}}\bigr)
      \Bigr]\times\frac{1}{4}\nonumber\\
    &=&\frac{B-C}{2}\bigl(\ket{\alpha}\ket{0}-\ket{0}\ket{\beta}\bigr).
  \end{eqnarray}
  The constraint ${\Omega_3^{\text{\rm ECS}(-)}}\ket{\phi^{\prime}}=\ket{\phi^{\prime}}$ leads to
  \begin{eqnarray}
    C=-B\,.
  \end{eqnarray}
  Therefore, the eigenstate of all the three measurement settings $\Omega_1^{\text{\rm ECS}(-)},\Omega_2^{\text{\rm ECS}(-)},\Omega_3^{\text{\rm ECS}(-)}$ is the entangled coherent state
  \begin{eqnarray}
    \ket{\phi^{\prime}}=\ket{{\psi^{\text{\rm ECS}(-)}_{\alpha,\beta}}} =
    \frac1{\sqrt{C_{\alpha,\beta}^{(-)}}}\Bigl(\ket{\alpha}\ket{0}-\ket{0}\ket{\beta}\Bigr).
  \end{eqnarray}
\end{proof}

Furthermore, we consider a class of two-mode entangled coherent state with the general form \cite{Sanders1992}
\begin{eqnarray}
  \ket{{\tilde{\psi}^{\text{\rm ECS}(-)}}}=\frac1{\sqrt{\tilde{C}^{(-)}}}\Bigl(\ket{\alpha_1}\ket{\alpha_2}-\ket{\beta_1}\ket{\beta_2}\Bigr),
\end{eqnarray}
which is locally equivalent to \ket{{\psi^{\text{\rm ECS}(-)}_{\alpha,\beta}}} whenever the constraint
\begin{eqnarray}
  \Im(\alpha_1\beta_1^{*})+\Im(\alpha_2\beta_2^{*})=2n\pi\,, \quad(n\in\mathbb{Z})
\end{eqnarray}
is satisfied.
The normalization is given by $\tilde{C}^{(-)}=2-\exp\bigl[-\frac{1}{2}(|\alpha_1|^2+|\beta_1|^2+|\alpha_2|^2+|\beta_2|^2)\bigr]\bigl(\E^{\alpha_1^{*}\beta_1+\alpha_2^{*}\beta_2}+\E^{\alpha_1\beta_1^{*}+\alpha_2\beta_2^{*}}\bigr)$.
Moreover, \ket{{\tilde{\psi}^{\text{\rm ECS}(-)}}} is also locally equivalent to \ket{{\psi^{\text{\rm ECS}}_{\alpha,\beta}}} whenever the constraint
\begin{eqnarray}
  \Im(\alpha_1\beta_1^{*})+\Im(\alpha_2\beta_2^{*})=(2n+1)\pi\,, \quad(n\in\mathbb{Z})
\end{eqnarray}
is satisfied.
Hence, we conclude that the state \ket{{\tilde{\psi}^{\text{\rm ECS}(-)}}} can be verified by $\Omega^{\text{\rm ECS}}$ or $\Omega^{\text{\rm ECS}(-)}$ under proper local displacement operations whenever we have
\begin{eqnarray}
  \Im(\alpha_1\beta_1^{*})+\Im(\alpha_2\beta_2^{*})=n\pi\,, \quad(n\in\mathbb{Z})\,.
\end{eqnarray}
Similarly, the constraint for verifying \ket{{\tilde{\psi}^{\text{\rm ECS}}}} as in Eq.~\eqref{eq:two-mode} can be relaxed as the same.

\section{Appendix E: Example of the optimization for $1/\nu_{\text{\rm opt}}$}\label{app:example}
In this section we show how to perform numerical optimizations to find the optimal efficiency of $1/\nu_{\text{\rm opt}}$.
Due to the infinite dimensional Hilbert space of the CV quantum states, we cannot consider all the possible noises as that in the DV scenario \cite{Pallister.etal2018}.
Instead, here we tackle the problem from a more practical perspective.

The most common noise in reality is the decoherence of a mode due to photon absorption losses, which can be described by a single parameter $\eta$ to represent the fraction of photons that survives the noisy channel, i.e.,
\begin{eqnarray}\label{eq:photon_loss}
  \ket{\alpha}\otimes\ket{0}_{E} \rightarrow \ket{\sqrt{\eta}\alpha}\otimes\ket{\sqrt{1-\eta}\alpha}_{E}\,,
\end{eqnarray}
where $\ket{\cdot}_E$ refers to the environment mode.
Then we make the reasonable assumption that (i): The output state is measured immediately after the displacement, so the photon loss occurs mainly before the displacement.
With this assumption, we have

\noindent{\bf Decoherent States Type A:}
The decoherent states caused by photon losses can be described as $\sigma_a=\ket{\phi_a}\bra{\phi_a}$ where
\begin{eqnarray}
  \ket{\phi_a}:=\ket{\alpha-(1-\sqrt{\eta})\alpha/2}\ket{(1-\sqrt{\eta})\alpha/2}+\ket{(1-\sqrt{\eta})\alpha/2}\ket{\alpha-(1-\sqrt{\eta})\alpha/2}\,,
\end{eqnarray}
in which the normalization is omitted.
The passing probability $\tr(\Omega_l\sigma_i)$ of the decoherent states depends on the transmission probability $\eta$, so does the infidelity $\epsilon=1-\bra{\psi}\sigma_i\ket{\psi}$.
Then we can use the parameter $\eta$ as a bridge to link the passing probability $\tr(\Omega_l\sigma_i)$ and the infidelity $\epsilon$.
Considering the case of $\alpha=1$, for the decoherent state $\sigma_a$ with three measurement settings, the parameters $k_{l,i}$ as defined in Eq.~\eqref{eq:linear} of Appendix A are given by
\begin{eqnarray}
  k_a=\{0.42(7),0.42(7),0.00(0)\}\,.
\end{eqnarray}
The numerical results show that ${\Omega_3^\text{\rm ECS}}$ is not necessary for verifying such a decoherent state as we always have $\tr({\Omega_3^\text{\rm ECS}}\sigma_a)=1$.
Moreover, we can confirm the symmetry of $\sigma_a$ with respect to the two settings ${\Omega_1^\text{\rm ECS}}$ and $\Omega_2^{\text{\rm ECS}}$.
Note that the numerical results presented here and below are all based on using the resolution PNRD$(5)$.

Another possible decoherence is from the imperfect construction of the states.
For instance, the construction decoherence from imperfect displacements can be described as
\begin{eqnarray}
  \ket{\alpha} \rightarrow \ket{\alpha+\Delta}\,,
\end{eqnarray}
where the small value $\Delta$ represents the displacement error.
Then we make another reasonable assumption that (ii): The construction decoherence mainly comes from displacement of the two modes.
With this assumption, we have

\noindent{\bf Decoherent States Type B:}
Considering the extreme case of one-mode error only, the decoherent states caused by displacement error can be described as $\sigma_e=\ket{\phi_e}\bra{\phi_e}$ where
\begin{eqnarray}
  \ket{\phi_e}:=\ket{\alpha+\Delta}\ket{0}+\ket{\Delta}\ket{\alpha}\,.
\end{eqnarray}
The numerical results are
\begin{eqnarray}
  k_e=\{0.31(5),0.31(5),0.99(3)\}\,.
\end{eqnarray}

\noindent{\bf Decoherent States Type C:}
Considering symmetric two-mode error, the decoherent states caused by displacement error can be described as $\sigma_s=\ket{\phi_s}\bra{\phi_s}$ where
\begin{eqnarray}
  \ket{\phi_s}:=\ket{\alpha+\Delta}\ket{\Delta}+\ket{\Delta}\ket{\alpha+\Delta}\,.
\end{eqnarray}
The numerical results are
\begin{eqnarray}
  k_s=\{0.39(9),0.33(1),0.99(2)\}\,.
\end{eqnarray}

Next, we move on to consider these two kinds of decoherence effects together.
With the above two reasonable assumptions, the decoherent states can be defined as
\begin{eqnarray}
  \ket{\phi_e(\eta)}&:=&\ket{\alpha-(1-\sqrt{\eta})\alpha/2+\Delta}\ket{(1-\sqrt{\eta})\alpha/2}+\ket{(1-\sqrt{\eta})\alpha/2+\Delta}\ket{\alpha-(1-\sqrt{\eta})\alpha/2}\,,\nonumber\\
  \ket{\phi_s(\eta)}&:=&\ket{\alpha-(1-\sqrt{\eta})\alpha/2+\Delta}\ket{(1-\sqrt{\eta})\alpha/2+\Delta}+\ket{(1-\sqrt{\eta})\alpha/2+\Delta}\ket{\alpha-(1-\sqrt{\eta})\alpha/2+\Delta}\,,
\end{eqnarray}
which are the noisy states in Eq.~\eqref{eq:noisyStates} of the main text used for demonstration.
Note that $\sigma_e=\sigma_e(100\%)$ and $\sigma_s=\sigma_s(100\%)$ are the cases if considering the displacement error only, and we relabel $\sigma_a=\sigma_{\Delta=0}$ for clarity which is the case if considering the photon loss only.
In the following we discuss various scenarios.

\noindent{\bf Case 1:}
If considering the photon loss and the displacement error independently, the decoherent states are given by
\begin{eqnarray}
  \{\sigma_i\}=\{\sigma_e(100\%), \sigma_s(100\%), \sigma_{\Delta=0}\}\,.
\end{eqnarray}
The parameters $k_{l,i}$ form the following $3\text{(number of settings)} \times 3\text{(types of noise)}$ matrix
\begin{eqnarray}
  K_A = \left[\begin{matrix}
    0.31(5)  &0.39(9)  &0.42(7) \\
    0.31(5)  &0.33(1)  &0.42(7) \\
    0.99(3)  &0.99(2)  &0.00(0)
  \end{matrix}\right]\!.
\end{eqnarray}
With this, we can optimize the QSV protocol as in Eq.~\eqref{eq:v_opt}.
The objective of the optimization is to find the probability distribution $\{\mu_l\}$ of the three measurement settings which makes the passing probability of the worst case to be minimal.
To be specific, we calculate the passing probability of all the possible decoherent states that we consider, and find the maximal one (i.e. the worst case) with certain $\{\mu_l\}$.
Then the probability distribution $\{\mu_l\}$ is varied in order to get the optimal protocol.
In this case the optimal efficiency is given by
${1/\nu_{\text{\rm opt}} \approx 2.60(4)}$ with
${\{\mu_l\}=\{0.49(7),0.40(2),0.10(1)\}}$.

\noindent{\bf Case 2:}
If considering the photon loss and the one-mode displacement error together, the decoherent states are given by
\begin{eqnarray}
  \{\sigma_i\}=\{\sigma_e(100\%), \sigma_e(99\%), \sigma_e(98\%), \sigma_{\Delta=0}\}\,.
\end{eqnarray}
The parameters $k_{l,i}$ form the following $3\text{(number of settings)} \times 4\text{(types of noise)}$ matrix
\begin{eqnarray}
  K_B = \left[\begin{matrix}
    0.31(5)  &0.33(4)  &0.35(3)  &0.42(7) \\
    0.31(5)  &0.29(4)  &0.27(2)  &0.42(7) \\
    0.99(3)  &0.99(2)  &0.99(0)  &0.00(0)
  \end{matrix}\right]\!.
\end{eqnarray}
Hence, the optimal efficiency is given by
${1/\nu_{\text{\rm opt}} \approx 2.60(4)}$ with
${\{\mu_l\}=\{0.49(7),0.40(2),0.10(1)\}}$.

\noindent{\bf Case 3:}
If considering the photon loss and the symmetric two-mode displacement error together, the decoherent states are given by
\begin{eqnarray}
  \{\sigma_i\}=\{\sigma_s(100\%), \sigma_s(99\%), \sigma_s(98\%), \sigma_{\Delta=0}\}\,.
\end{eqnarray}
The parameters $k_{l,i}$ form the following $3\text{(number of settings)} \times 4\text{(types of noise)}$ matrix
\begin{eqnarray}
  K_C = \left[\begin{matrix}
    0.39(9)  &0.42(2)  &0.44(5)  &0.42(7) \\
    0.33(1)  &0.30(9)  &0.28(8)  &0.42(7) \\
    0.99(2)  &0.99(1)  &0.98(9)  &0.00(0)
  \end{matrix}\right]\!.
\end{eqnarray}
Hence, the optimal efficiency is given by
${1/\nu_{\text{\rm opt}} \approx 2.40(7)}$ with
${\{\mu_l\}=\{0.97(2),0.00(0),0.02(8)\}}$.

\noindent{\bf Case 4:}
If considering the photon loss and two kinds of displacement errors altogether, the decoherent states are given by
\begin{eqnarray}\label{eq:all_noise}
  \{\sigma_i\}=\{\sigma_e(100\%), \sigma_e(99\%), \sigma_e(98\%), \sigma_s(100\%), \sigma_s(99\%), \sigma_s(98\%), \sigma_{\Delta=0}\}\,.
\end{eqnarray}
The parameters $k_{l,i}$ form the following $3\text{(number of settings)} \times 7\text{(types of noise)}$ matrix
\begin{eqnarray}
  K_D = \left[\begin{matrix}
    0.31(5)  &0.33(4)  &0.35(3)  &0.39(9)  &0.42(2)  &0.44(5)  &0.42(7) \\
    0.31(5)  &0.29(4)  &0.27(2)  &0.33(1)  &0.30(9)  &0.28(8)  &0.42(7) \\
    0.99(3)  &0.99(2)  &0.99(0)  &0.99(2)  &0.99(1)  &0.98(9)  &0.00(0)
  \end{matrix}\right]\!.
\end{eqnarray}
Hence, the optimal efficiency is given by
${1/\nu_{\text{\rm opt}} \approx 2.60(4)}$ with
${\{\mu_l\}=\{0.49(7),0.40(2),0.10(1)\}}$. See the results shown in Fig.~\ref{fig:fit_linear}

\begin{figure}[t]
  \includegraphics[width=0.8\columnwidth]{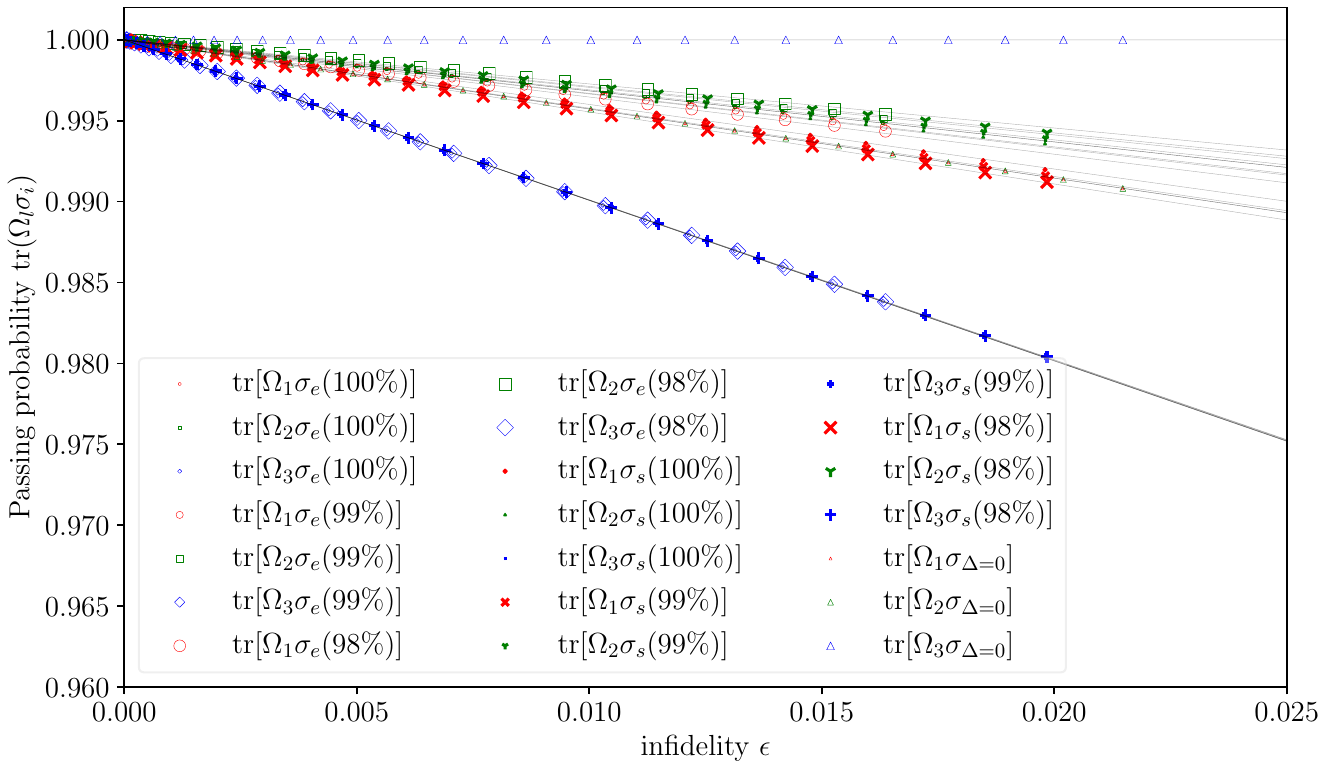}
  \caption{Numerical results showing the linear relationship between the passing probability $\tr(\Omega_l\sigma_i)$ and the infidelity $\epsilon$ by considering  different types of noisy states $\sigma_i$ as in Eq.~\eqref{eq:all_noise}.}
  \label{fig:fit_linear}
\end{figure}

Some remarks regarding the numerical process are in order.
First, the grey lines shown in Fig.~\ref{fig:fit_linear} are fitted with the function $y=1-kx$.
Although this linear relationship is proved only with the noisy states as defined in Eq.~\eqref{eq:def_noise} of Appendix A, all the decoherent states that we discuss here approximately satisfy this relationship, especially for the high-fidelity scenarios that we are interested in.
This property greatly simplifies the numerical processes.
Second, the examples we consider here cover several important decoherence effects only, thus more considerations shall give better results.
However, the optimization defined in Eq.~\eqref{eq:v_opt} of the main text (also in Appendix A) may not be applicable.
Nevertheless, the link between the parameter used to describe the decoherent states and the infidelity $\epsilon$ always exists, hence one can still perform the numerical optimization with more technical efforts.
Last, there is an interesting observation of the numerical results such that the independent consideration of various possible decoherent effects is able to get the worst case.
This observation is also helpful for the optimization process.

Finally, we do a comparison between our protocol and a real experiment reported in 
Ref.~\cite{Israel.etal2019} where an ECS in the form \ket{\psi^{\text{\rm ECS}}_{\alpha,\alpha}} was constructed.
The fidelity of the ECS with $\alpha\approx0.515$ was about $99.99\%$ by using quantum tomography.
The data was collected more than 24 hours with a Ti:sapphire oscillator operated at 80MHz, and 
the number of measurements used was more than $10^{10}$.
For the same state, instead, our protocol gives the optimal efficiency of ${1/\nu_{\text{\rm opt}} \approx 2.96(6)}$ with
${\{\mu_l\}=\{0.57(4),0.14(2),0.28(5)\}}$.
For the $99.99\%$ fidelity, it requires $1.36(6)\times10^5$ measurements to reach the confidence level of $99\%$ which is much more efficient.

\section{Appendix F: Proof of Theorem~\ref{thm:GHZ}}\label{app:GHZ}
Similar to the proof of Theorem~\ref{thm:ECS}, we prove Theorem~\ref{thm:GHZ} by showing that
\begin{itemize}
    \item[(1)] The eigenstate \ket{\phi} of all the $2(m\!-\!1)$ measurement settings as defined in Eq.~\eqref{eq:GHZm_1} is the superposition of the states $\bigotimes_{i=1}^{m}\ket{\alpha_i}$ and $\ket{0}^{\otimes m}$;
    \item[(2)] Such a state must be \ket{\psi^{\text{\rm GHZ-m}}} if it also satisfies $\Omega_{2m-1}^{\text{\rm GHZ-m}}\ket{\phi}=\ket{\phi}$, where the measurement setting is defined in Eq.~\eqref{eq:GHZm_2}.
\end{itemize}

As shown in Appendix~B, the two-mode GHZ-like state \ket{\psi^{\text{\rm GHZ-2}}} is locally equivalent to \ket{\psi^{\text{\rm ECS}}_{\alpha_1,-\alpha_2}} with the displacement $\openone\otimes D(-\alpha_2)$.
Then, we have the following two lemmas directly.
\begin{lemma}\label{lem:2mode}
  The eigenstates of $\Omega_{2l-1}^{\text{\rm GHZ-m}}$ and $\Omega_{2l}^{\text{\rm GHZ-m}}$ $(\forall l=1,2,\cdots,m\!-\!1)$ are the superposition
  \begin{eqnarray}
    \ket{\phi_l}=x\ket{\alpha_{l}}\ket{\alpha_{l+1}}\ket{\psi_1}+y\ket{0}\ket{0}\ket{\psi_2}\,,
  \end{eqnarray}
  where $x,y$ are the normalized superposition amplitudes, and \ket{\psi_1} and \ket{\psi_2} are the states in the subspace formed without the two modes $l$ and $l+1$.
\end{lemma}
\begin{lemma}
  The eigenstates of $\Omega_{2l-1}^{\text{\rm GHZ-m}}$, $\Omega_{2l}^{\text{\rm GHZ-m}}$, $\Omega_{2l+1}^{\text{\rm GHZ-m}}$, and $\Omega_{2l+2}^{\text{\rm GHZ-m}}(\forall l=1,2,\cdots,m\!-\!2)$ are the superposition
  \begin{eqnarray}
    \ket{\phi_l}=x\ket{\alpha_{l}}\ket{\alpha_{l+1}}\ket{\alpha_{l+2}}\ket{\psi_1}+y\ket{0}\ket{0}\ket{0}\ket{\psi_2}\,,
  \end{eqnarray}
  where $x,y$ are the normalized superposition amplitudes, and \ket{\psi_1} and \ket{\psi_2} are the states in the subspace formed without the three modes $l,l+1$, and $l+2$.
\end{lemma}
\begin{proof}
  For clarity, we assign the labels $\alpha,\beta,\gamma$ to the modes $\alpha_l,\alpha_{l+1},\alpha_{l+2}$ respectively.
  Following Lemma~\ref{lem:2mode}, we have the eigenstates for $\Bigl\{\Omega_{2l-1}^{\text{\rm GHZ-m}},\Omega_{2l}^{\text{\rm GHZ-m}}\Bigr\}$ and $\Bigl\{\Omega_{2(l+1)-1}^{\text{\rm GHZ-m}},\Omega_{2(l+1)}^{\text{\rm GHZ-m}}\Bigr\}$ given by
  \begin{eqnarray}
    \ket{\phi_1} &=& \sum_k A_k \ket{\alpha}\ket{\beta}\ket{k} + \sum_k B_k \ket{0}\ket{0}\ket{k}\,,\\
    \ket{\phi_2} &=& \sum_k C_k \ket{\alpha}\ket{k}\ket{\gamma} + \sum_k D_k \ket{0}\ket{k}\ket{0}\,,
  \end{eqnarray}
  where $A_k,B_k,C_k,D_k$ are the normalized superposition amplitudes.
  In order to get the eigenstate of all the four measurement settings, we let $\ket{\phi_1}=\ket{\phi_2}$
  and expand them under the Fock bases
  \begin{eqnarray}
    &&\sum_{n,m,p}
    A_p \braket{n}{\alpha} \braket{m}{\beta} \ket{n}\ket{m}\ket{p}
    + \sum_{p}
    B_p  \ket{0}\ket{0}\ket{p}\nonumber\\
    &=&
    \sum_{n,m,p}
    C_m \braket{n}{\alpha} \braket{p}{\gamma} \ket{n}\ket{m}\ket{p}
    + \sum_{m}
    D_m  \ket{0}\ket{m}\ket{0}\,.
  \end{eqnarray}
  Note that the same terms should have the same superposition amplitudes.
  To be specific, one has
  \begin{eqnarray}
  \begin{array}{lll}
    \left.
    \begin{array}{ll}
      \ket{0}\ket{0}\ket{p}\,, &(p\neq0)\\
      \ket{n}\ket{0}\ket{p}\,, &(n\neq0,p\neq0)
    \end{array}
    \right\}
    &\Rightarrow&
    B_p=0\,,\quad(p\neq0)\,,\\
    \left.
    \begin{array}{ll}
      \ket{0}\ket{m}\ket{0}\,, &(m\neq0)\\
      \ket{n}\ket{m}\ket{0}\,, &(n\neq0,m\neq0)
    \end{array}
    \right\}
    &\Rightarrow&
    D_m=0\,,\quad(m\neq0)\,,\\
    \left.
    \begin{array}{ll}
      \ket{n}\ket{m}\ket{p}\,, &(m\neq0,p\neq0)\\
      \ket{n}\ket{0}\ket{0}\,, &(n\neq0)
    \end{array}
    \right\}
    &\Rightarrow&
    A_p\braket{m}{\beta}=C_m\braket{p}{\gamma}\,,\quad(\forall m,p)\,,\\
    \left.
    \begin{array}{ll}
      \ket{0}\ket{0}\ket{0}\,, &\\
      \ket{n}\ket{0}\ket{0}\,, &(n\neq0)
    \end{array}
    \right\}
    &\Rightarrow&
    B_0=D_0\,.
  \end{array}
  \end{eqnarray}
  By defining $A_p/\braket{p}{\gamma}=C_m/\braket{m}{\beta}=T$ and $B_0=D_0=K$, we have
  \begin{eqnarray}
    \ket{\phi}
    = \ket{\phi_1}
    &=& \sum_{p=0}^\infty T\braket{p}{\gamma}\ket{\alpha}\ket{\beta}\ket{p} + K \ket{0}\ket{0}\ket{0}
    = T\ket{\alpha}\ket{\beta}\ket{\gamma} + K \ket{0}\ket{0}\ket{0}\nonumber\\
    &=& \sum_{m=0}^\infty T\braket{m}{\beta}\ket{\alpha}\ket{m}\ket{\gamma} + K \ket{0}\ket{0}\ket{0}
    =\ket{\phi_2}\,,
  \end{eqnarray}
  for the subspace $l,l+1,l+2$.
  Finally, we conclude that the eigenstate can be written as
  \begin{eqnarray}
    \ket{\phi}=T\ket{\alpha}\ket{\beta}\ket{\gamma}\ket{\psi_1}+K\ket{0}\ket{0}\ket{0}\ket{\psi_2}\,.
  \end{eqnarray}
\end{proof}

Therefore, we deduce that the eigenstates of all the $2(m\!-\!1)$ measurement settings of the verification protocol $\Omega^{\text{\rm GHZ-m}}$ are the superposition of the states $\bigotimes_{i=1}^{m}\ket{\alpha_i}$ and $\ket{0}^{\otimes m}$, of which we write as
\begin{eqnarray}
  \ket{\phi}= x\bigotimes_{i=1}^{m}\ket{\alpha_i} + y\ket{0}^{\otimes m}\,.
\end{eqnarray}
Then, we consider the last measurement setting $\Omega^{\text{\rm GHZ-m}}$, in which $\bigl(\pi^{\otimes m}\bigr)^{+}$ is the projector onto the eigenspace with eigenvalue 1 of the operator $\pi^{\otimes m}$.
Note that
\begin{eqnarray}
  \pi^{\pm}\ket{\alpha} &=& \frac{\sqrt{C_{\pm}}}{2}\ket{\pm_\alpha}=\frac{\ket{\alpha}\pm\ket{-\alpha}}{2}\,.
\end{eqnarray}
Considering the symmetry of $\bigl(\pi^{\otimes m}\bigr)^{+}$ that
\begin{eqnarray}
  \bigl(\pi^{\otimes m}\bigr)^{+}=\sum_{\substack{t_i\in\{0,1\}\\\sum t_i=\text{even}}} \bigotimes_{i=1}^{m} \pi^{t_i}\,,\quad(\text{here }\pi^{0}=\pi^{+},\pi^{1}=\pi^{-}\text{ for clarity})\,,
\end{eqnarray}
direct calculation shows
\begin{eqnarray}
  \Omega_{2m-1}^{\text{\rm GHZ-m}}\ket{\phi}
  &=& \biggl[\bigotimes_{i=1}^{m}D^{\dagger}(-\frac{\alpha_i}{2})\biggr]
  \bigl(\pi^{\otimes m}\bigr)^{+}
  \biggl[x\bigotimes_{i=1}^{m}\ket{\frac{\alpha_i}{2}}
  +y\bigotimes_{i=1}^{m}\ket{\frac{-\alpha_i}{2}}\biggr]
  \nonumber\\
  &=& \frac{x}{2^m}
  \sum_{ \substack{t_i\in\{0,1\}\\\sum t_i=\text{even}}}
  \bigotimes_{i=1}^{m}\Bigl[\ket{\alpha_i}+(-1)^{t_i}\ket{0}\Bigr]
  +\frac{y}{2^m}
  \sum_{ \substack{t_i\in\{0,1\}\\\sum t_i=\text{even}}}
  \bigotimes_{i=1}^{m}\Bigl[\ket{0}+(-1)^{t_i}\ket{\alpha_i}\Bigr]\nonumber\\
  &=& \frac{x}{2^m}
  \sum_{ \substack{t_i\in\{0,1\}\\\sum t_i=\text{even}}}
  \bigotimes_{i=1}^{m}\Bigl[\ket{\alpha_i}+(-1)^{t_i}\ket{0}\Bigr]
  +\frac{y}{2^m}
  \sum_{ \substack{t_i\in\{0,1\}\\\sum t_i=\text{even}}}
  (-1)^{\sum t_i}
  \bigotimes_{i=1}^{m}\Bigl[\ket{\alpha_i}+(-1)^{t_i}\ket{0}\Bigr]\nonumber\\
  &=& \frac{x+y}{2^m}
  \sum_{ \substack{t_i\in\{0,1\}\\\sum t_i=\text{even}}}
  \bigotimes_{i=1}^{m}\Bigl[\ket{\alpha_i}+(-1)^{t_i}\ket{0}\Bigr]\nonumber\\
  &=& \frac{x+y}{2} \Biggl(\bigotimes_{i=1}^{m}\ket{\alpha_i}+\ket{0}^{\otimes m}\Biggr).
\end{eqnarray}
With the constraint $\Omega_{2m-1}^{\text{\rm GHZ-m}}\ket{\phi}=\ket{\phi}$, we have $x=y$.
Then the eigenstate of all the measurement settings of the protocol $\Omega^{\text{\rm GHZ-m}}$ must be the coherent GHZ-like state
\begin{eqnarray}
  \ket{\phi} = \ket{\psi^{\text{\rm GHZ-}m}}=\frac1{\sqrt{C}}\Biggl(\bigotimes_{i=1}^{m}\ket{\alpha_i}+\ket{0}^{\otimes m}\Biggr).
\end{eqnarray}

\section{Appendix G: Verification of a different kind of GHZ-like state and its generalization}
Consider a different class of multimode entangled coherent state
\begin{eqnarray}\label{eq:GHZ2}
  \ket{{\psi^{\prime}}^{\text{\rm GHZ-}m}}=\frac1{\sqrt{C^{\prime}}}\Biggl(\bigotimes_{i=1}^{m}\ket{\alpha_i}-\ket{0}^{\otimes m}\Biggr),
\end{eqnarray}
where the normalization is $C^{\prime}=2\!\left[1-\E^{-\sum_{i}|\alpha_i|^2/2}\right]$.
Similar to the verification of \ket{{\psi}^{\text{\rm GHZ-}m}}, it is sufficient to verify \ket{{\psi^{\prime}}^{\text{\rm GHZ-}m}} with $(2m\!-\!1)$ measurement settings that
\begin{eqnarray}
  {\Omega'}_{2l-1}^{\text{\rm GHZ-}m}&=&\Omega_{2l-1}^{\text{\rm GHZ-}m}=
  \mathcal{P}_{l}\biggl\{\!
  \Bigl[B_{2l\!-\!1}^\dagger(\openone \!- \tau^{-}\!\otimes\tau^{-}) B_{2l\!-\!1}\Bigr]\!
  \otimes\openone^{\otimes (m-2)}\!
  \biggr\},\nonumber\\
  {\Omega'}_{2l}^{\text{\rm GHZ-}m}&=&\Omega_{2l}^{\text{\rm GHZ-}m}=
  \mathcal{P}_{l}\biggl\{\!
  \Bigl[B_{2l}^\dagger(\openone \!- \tau^{-}\!\otimes\tau^{-}) B_{2l}\Bigr]\!
  \otimes\openone^{\otimes (m-2)}\!
  \biggr\},\\
  {\Omega'}_{2m-1}^{\text{\rm GHZ-}m}&=&\Biggl[\bigotimes_{i=1}^{m}D^{\dagger}(-\frac{\alpha_i}{2})\Biggr]
  \bigl(\pi^{\otimes m}\bigr)^{-}
  \Biggl[\bigotimes_{i=1}^{m}D(-\frac{\alpha_i}{2})\Biggr],\nonumber
\end{eqnarray}
with $l=1,2,\cdots,m$, where $B_{2l-1}=D(-\alpha_l)\otimes\openone$ and $B_{2l}= \openone\otimes D(-\alpha_{l+1})$ are local operations.
The symbol $\mathcal{P}_{l}$ indicates the permutation that only the $l$ and $l+1$ modes are operated on for each setting.
See the following theorem with its proof.
\begin{theorem}\label{thm:GHZ2}
  The $m$-mode GHZ-like coherent states \ket{{\psi'}^{\text{\rm GHZ-m}}} can be verified efficiently by the protocol
  \begin{eqnarray}
    {\Omega'}^{\text{\rm GHZ-}m}=\sum_{l=1}^{2m-1}\mu_l{\Omega'}_l^{\text{\rm GHZ-}m}\,,
  \end{eqnarray}
  where the probability distribution $\{\mu_l\}$ is arbitrary.
  An optimal verification efficiency can be obtained by optimizing $\{\mu_l\}$ under specific scenarios as constrained by Eq.~\eqref{eq:v_opt}.
\end{theorem}
\begin{proof}
  Same as the proof in Appendix F, the eigenstates of all the $2(m\!-\!1)$ measurement settings ${\Omega'}_{2l-1}^{\text{\rm GHZ-}m}$ and ${\Omega'}_{2l}^{\text{\rm GHZ-}m}$ take the form
  \begin{eqnarray}
    \ket{\phi}= x\bigotimes_{i=1}^{m}\ket{\alpha_i} + y\ket{0}^{\otimes m}\,.
  \end{eqnarray}
  Considering the symmetry of the parity measurement
  \begin{eqnarray}
    \bigl(\pi^{\otimes m}\bigr)^{-}=\sum_{\substack{t_i\in\{0,1\}\\\sum t_i=\text{odd}}} \bigotimes_{i=1}^{m} \pi^{t_i}\,,\quad(\text{here }\pi^{0}=\pi^{+},\pi^{1}=\pi^{-}\text{ for clarity})\,,
  \end{eqnarray}
  we have
  \begin{eqnarray}
    {\Omega'}_{2m-1}^{\text{\rm GHZ-m}}\ket{\phi}
    &=& \biggl[\bigotimes_{i=1}^{m}D^{\dagger}(-\frac{\alpha_i}{2})\biggr]
    \bigl(\pi^{\otimes m}\bigr)^{-}
    \biggl[x\bigotimes_{i=1}^{m}\ket{\frac{\alpha_i}{2}}
    +y\bigotimes_{i=1}^{m}\ket{\frac{-\alpha_i}{2}}\biggr]
    \nonumber\\
    &=& \frac{x}{2^m}
    \sum_{ \substack{t_i\in\{0,1\}\\\sum t_i=\text{odd}}}
    \bigotimes_{i=1}^{m}\Bigl[\ket{\alpha_i}+(-1)^{t_i}\ket{0}\Bigr]
    +\frac{y}{2^m}
    \sum_{ \substack{t_i\in\{0,1\}\\\sum t_i=\text{odd}}}
    (-1)^{\sum t_i}
    \bigotimes_{i=1}^{m}\Bigl[\ket{\alpha_i}+(-1)^{t_i}\ket{0}\Bigr]\nonumber\\
    &=& \frac{x-y}{2^m}
    \sum_{ \substack{t_i\in\{0,1\}\\\sum t_i=\text{odd}}}
    \bigotimes_{i=1}^{m}\Bigl[\ket{\alpha_i}+(-1)^{t_i}\ket{0}\Bigr]\nonumber\\
    &=& \frac{x-y}{2} \Biggl(\bigotimes_{i=1}^{m}\ket{\alpha_i}-\ket{0}^{\otimes m}\Biggr).
  \end{eqnarray}
  With the constraint ${\Omega'}_{2m-1}^{\text{\rm GHZ-m}}\ket{\phi}=\ket{\phi}$, we have $y=-x$.
  Hence, the eigenstate of all the $(2m\!-\!1)$ measurement settings is the multimode entangled coherent state
  \begin{eqnarray}
    \ket{\phi}=\ket{{\psi^{\prime}}^{\text{\rm GHZ-}m}}
    =\frac1{\sqrt{C^{\prime}}}\Biggl(\bigotimes_{i=1}^{m}\ket{\alpha_i}-\ket{0}^{\otimes m}\Biggr).
  \end{eqnarray}
\end{proof}

Furthermore, a class of multimode entangled coherent state with the general form
\begin{eqnarray}\label{eq:GGHZ2}
  \ket{\Tilde{\psi'}^{\text{\rm GHZ-}m}}
  =\frac1{\sqrt{\tilde{C'}}}\Biggl(\bigotimes_{i=1}^{m}\ket{\alpha_i}-\bigotimes_{i=1}^{m}\ket{\beta_i}\Biggr)
\end{eqnarray}
is locally equivalent to \ket{{\psi'}^{\text{\rm GHZ-}m}} whenever the constraint
\begin{eqnarray}
  \sum_i\Im(\alpha_i\beta_i^{*})=2n\pi\,, \quad(n\in\mathbb{Z})
\end{eqnarray}
is satisfied.
The normalization is given by $\tilde{C'}=2-\exp\bigl[-\sum(|\alpha_i|^2+|\beta_i|^2)/2\bigr]\bigl[\exp(\sum\alpha_i^{*}\beta_i)+\exp(\sum\alpha_i\beta_i^{*})\bigr]$.
Moreover, \ket{\Tilde{\psi'}^{\text{\rm GHZ-}m}} is locally equivalent to \ket{{\psi}^{\text{\rm GHZ-}m}} whenever the constraint
\begin{eqnarray}
  \sum_i\Im(\alpha_i\beta_i^{*})=(2n+1)\pi\,, \quad(n\in\mathbb{Z})
\end{eqnarray}
is satisfied.
Hence, we conclude that the state \ket{\Tilde{\psi'}^{\text{\rm GHZ-}m}} can be verified by ${\Omega}^{\text{\rm GHZ-}m}$ or ${\Omega'}^{\text{\rm GHZ-}m}$ under proper local displacement operations whenever we have
\begin{eqnarray}
  \sum_i\Im(\alpha_i\beta_i^{*})=n\pi\,, \quad(n\in\mathbb{Z})\,.
\end{eqnarray}
Similarly, the constraints for verifying \ket{\Tilde{\psi}^{\text{\rm GHZ-}m}} as in Eq.~\eqref{eq:GGHZ} can be relaxed as the same.

\end{document}